\documentclass[12pt,a4paper]{article}
\usepackage[top=0.8in, bottom=0.8in, left=0.8in, right=0.8in]{geometry}
\usepackage{setspace}
\usepackage{url} % not crucial - just used below for the URL
\usepackage{times}
\usepackage{paralist}
\usepackage{natbib, xcolor, dsfont}
\usepackage{subcaption}
\usepackage{comment}

\usepackage[font=small]{caption}
\definecolor{darkblue}{rgb}{0,0,.6}
\usepackage[pdftex,colorlinks=true,unicode]{hyperref}
\hypersetup{citecolor=darkblue,linkcolor=darkblue,urlcolor=darkblue}
\usepackage{amssymb}
\usepackage{latexsym}
\usepackage{amsfonts}
\usepackage{amsthm}
\usepackage{amsmath, microtype}
\usepackage{graphicx}
\usepackage{enumitem}
\usepackage{bbm}
\newcommand\figref{Figure~\ref}
\newcommand{\convergeInP}{\stackrel{i.p.}{\longrightarrow}}

\usepackage[normalem]{ulem}

\usepackage[font={small,it}, labelfont={bf}]{caption}

\newtheorem{definition}{Definition}
\newtheorem{theorem}{Theorem}

\newtheorem{innercustomgeneric}{\customgenericname}
\providecommand{\customgenericname}{}
\newcommand{\newcustomtheorem}[2]{%
  \newenvironment{#1}[1]
  {%
   \renewcommand\customgenericname{#2}%
   \renewcommand\theinnercustomgeneric{##1}%
   \innercustomgeneric
  }
  {\endinnercustomgeneric}
}
\newcustomtheorem{customthm}{Theorem}
\newcustomtheorem{customlemma}{Lemma}

\theoremstyle{definition}

\newtheorem{assumption}{Assumption}
\graphicspath{{plots/}}
\usepackage{titlesec}
\titleformat{\section}{\normalfont\Large\bfseries}{\thesection}{1em}{}

\usepackage{graphics}
\usepackage{authblk}
\newcommand{\Rlogo}{\protect\includegraphics[height=1.8ex,keepaspectratio]{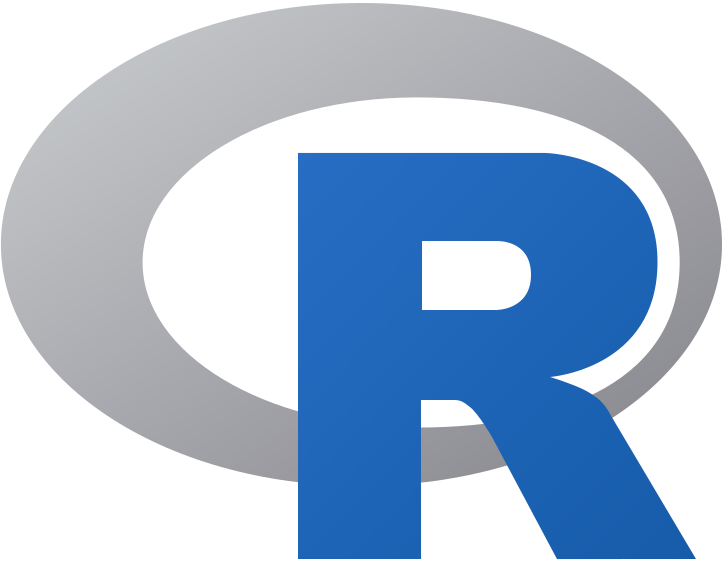}}

\begin{document}

\title{\Large\bf Double Descent in High-dimensional Linear Discriminant Analysis}
\author[a]{Yonghe Lu}
\author[b]{Han Lin Shang}
\author[a]{Yanrong Yang}
\author[a]{Kehan Zhao}
	
\affil[a]{The Australian National University}
\affil[b]{Macquarie University}

\date{}

\maketitle
\onehalfspacing

\begin{abstract}
The double descent phenomenon observed in deep neural networks has called into question conventional perspectives on model complexity in machine learning, particularly the classical U-shaped bias–variance trade-off curve. While this phenomenon has been extensively studied in modern high-capacity models, relatively little attention has been paid to its implications for classical statistical methods such as linear discriminant analysis (LDA), a long-standing classification technique that seeks a linear combination of features that best separates two classes. We leverage random matrix theory to characterize the asymptotic misclassification error of LDA as the dimension $p$ and sample size $n$ grow at the same rate $p/n\to\gamma$. Our main contribution is an asymptotic limit for the LDA risk: for a general population covariance matrix when $\gamma\in(0,1)$, and for the pseudo-inverse classifier with isotropic covariance when $\gamma\in(1,\infty)$. Together, these give an explicit form for the full double-descent risk curve of LDA in both the under- and over-parameterized regimes. We validate our theoretical findings through simulation studies and an empirical analysis of the ARCENE cancer classification dataset.
\end{abstract}


\section{Introduction}\label{section1}

Conventional statistical theory typically considers settings in which the sample size $n$ is large relative to the number of features $p$. However, advances in data collection and computing have produced increasingly high-dimensional data sets, in which $p$ may be comparable to or larger than $n$ \citep{Iain2009}. For example, high dimensionality is prominent in bio-informatics, where microarray technology enables probes of an entire genome to be placed on a chip \citep{xing2001}. In facial recognition, the number of pixels in an image is often much larger than the number of images, which requires modification of classification methods \citep{Yu2001}. These developments have motivated substantial research on the behavior of statistical procedures in high dimensions. In particular, recent studies challenge the classical U-shaped bias--variance trade-off by showing that prediction error may decrease again after the number of predictors exceeds the sample size. This phenomenon is referred to as ``benign overfitting" or ``double descent" \citep[see, e.g.,][]{Belkin2019, Hastie2019}.

Although the double descent phenomenon has been extensively investigated in high-dimensional regression settings, theoretical characterizations of its impact on classification procedures remain comparatively limited. In this paper, we examine double descent behavior in linear discriminant analysis (LDA) by deriving the high-dimensional asymptotic misclassification error under the homoscedastic Gaussian model, in both the under- and over-parameterized regimes.

Our analysis is conducted in a regime in which both the dimension $p$ and the sample size $n$ diverge, with their ratio satisfying
$p/n\rightarrow\gamma$. To obtain explicit deterministic equivalents for the classification error, we leverage tools from random matrix theory (RMT), which enable precise characterization of the spectral functionals of large random covariance matrices in high-dimensional asymptotics. The resulting limits describe how the LDA risk depends on the aspect ratio $\gamma$ and the underlying covariance structure. The theoretical findings are corroborated through a series of Monte Carlo simulations and empirical analyses designed to assess the performance of LDA across a range of covariance models, demonstrating close agreement between the asymptotic predictions and finite-sample behavior.

\subsection{Contributions}\label{sec:1.1}

To make the scope of our contribution precise, we summarize it as follows.
\begin{enumerate}[label=(\arabic*), itemsep=0pt, parsep=0pt]
\item For the under-parameterized regime $\gamma\in(0,1)$, we derive a closed-form deterministic equivalent for the LDA misclassification error under a general deterministic positive-definite covariance matrix $\boldsymbol{\Sigma}$; see Theorem~\ref{Thm1}. The limit depends on the population model through the Mahalanobis distance $\Delta$ and the dimensional ratio $\gamma$, and admits an interpretable decomposition into the effects of mean estimation and covariance estimation.

\item For the over-parameterized regime $\gamma\in(1,\infty)$, we derive a closed-form deterministic equivalent for the misclassification error of the Moore--Penrose pseudoinverse LDA classifier under isotropic covariance $\boldsymbol{\Sigma}=\sigma^2\mathbf I_p$; see Theorem~\ref{Thm2}. Together with Theorem~\ref{Thm1}, this result provides a complete asymptotic characterization of the double-descent risk curve across the under- and over-parameterized regimes in the common isotropic setting. For illustration, Figure~\ref{fig:error_3D} presents the resulting three-dimensional risk surface as a function of the dimensional ratio $\gamma$ and the Mahalanobis distance $\Delta$.

\item Through extensive simulations covering different covariance structures, non-Gaussian feature distributions, label and feature noise, and redundant features, together with an empirical analysis of the ARCENE dataset, we examine when double-descent behavior emerges in LDA. These experiments also assess whether the observed patterns persist beyond the Gaussian and covariance settings covered by the formal theory.
\end{enumerate}

\subsection{Related Literature}\label{sec:1.2}

We organize the relevant literature into two strands: classical and high-dimensional analyses of LDA, and the recent double-descent and benign-overfitting literature.

The error rate of discriminant analysis and its estimators is surveyed by \cite{McLachlan1992}; more recent work has examined LDA and related classifiers in high dimensions, including error-estimator behavior \citep{zollanvari2011}, variants of discriminant rule \citep{Saranadasa1993}, quadratic discriminant analysis \citep{cheng2004}, and other classification criteria \citep{li2016}. Closely related to ours, \citet{raudys1998} analyzed a pseudo-inverse classifier and observed that the expected error first decreases, rises as the dimension approaches the sample size, and decreases again under identity covariance. This anticipates the behavior we study, but \citet{raudys1998} use a fixed-dimension approximation, whereas we obtain closed-form random-matrix limits in the proportional regime that hold on both sides of the peak.

Closest to our proportional-growth setting are \citet{Wang2018} and \citet{cheng2022}, both of which derive asymptotic LDA error formulas under Gaussian sampling using Wishart-based arguments. \citet{Wang2018} focus on the under-parameterized regime $\gamma<1$ and consider settings in which either the covariance matrix or the class means are known, with particular emphasis on bias correction for regularized LDA. \citet{cheng2022} also study the over-parameterized regime, but their analysis focuses on unequal class proportions. Their primary concern is therefore the effect of class imbalance, whereas our analysis isolates the effects of covariance estimation and pseudoinversion under balanced classes.

Our contribution is an explicit and interpretable error formula that separates the mean- and covariance-estimation penalties and describes the risk continuously in both the under- and over-parameterized regimes.

The literature of double-descent and benign-overfitting provides the background for our analysis. It is empirically observed that a model's behavior in the over-parameterized region is at odds with the classical bias--variance trade-off, which predicts the commonly observed $U$-shaped curve when plotting the prediction error against a measure of model complexity. This phenomenon has been widely studied in regression problems \citep[see, e.g.,][]{Belkin2019, Hastie2019, Bartlett2020, nakkiran2020optimal, muthukumar2020harmless}.

In particular, \cite{Belkin2019} first coined the term ``double descent'' to describe the shape of the error curve and showed that the test error exhibits a second descent beyond the interpolation point for certain unregularized models. \cite{Hastie2019} provides a theoretical explanation for the minimum-norm least-squares estimator in generalized linear regression, deriving the prediction risk in the proportional regime $p/n \to \gamma$ and showing that the risk peaks as $\gamma \to 1$ before descending to a global minimum at $\gamma>1$. \cite{Bartlett2020} also studied the minimum-norm interpolator and, through matching risk bounds, identified many low-variance and unimportant directions in parameter space as essential for benign overfitting. 

Theoretical results on double descent for classifiers are scarcer than for regression, owing to the lack of closed-form error expressions. \cite{Montanari2019} analyzed the maximum-margin linear classifier and showed that the test error is monotonically decreasing in the over-parameterized region under conditions on the covariance spectrum. \cite{chatterji2021} proved finite-sample error bounds for the same classifier that approach Bayes risk as the dimension grows, while \cite{deng2021} studied gradient descent on logistic loss and \cite{wang2022binary} emphasized the interplay between performance, signal-to-noise ratio, and covariance structure. More recently, \cite{HVZ26} established the universality of benign overfitting for maximum-margin classifiers, relaxing the strong distributional assumptions of earlier mixture-model analyses. These analyses characterize the error through bounds, convergence directions, or hyperplane geometry. For the LDA plug-in classifier, by contrast, the misclassification error admits an explicit closed form, which allows us to trace the entire double-descent curve analytically and to separate the contributions of mean and covariance estimation.

One of the first demonstrations of double descent for classification uses Fisher linear discriminant analysis with the pseudo-inverse by \cite{Duin2000}, who shows a peak in generalization error when $p \simeq n$ and a reduction in error as the dimension keeps increasing, but only outlines the simulated curve without a theoretical account. Theorems~\ref{Thm1} and~\ref{Thm2} provide this missing account and extend the single curve of \cite{Duin2000} to a range of covariance structures.

\subsection{Organization}\label{sec:1.3}

The remainder of the paper is organized as follows. Section~\ref{sec:setup} introduces the problem together with the closed-form LDA misclassification-error expression. Section~\ref{sec:results} presents the main theoretical results and their interpretation. Section~\ref{sec:empirics} shows simulation results of the LDA error rate against $\gamma$ under different settings. Section~\ref{section6} conducts an empirical data analysis on the ARCENE cancer-classification dataset. Section~\ref{sec:conclusion} concludes with ideas on how the methodology can be extended.

\section{LDA Setup}\label{sec:setup}

We consider binary linear discriminant analysis with observations $(\mathbf{x}_i,y_i)$, $i=1,\dots,n$, where $\mathbf{x}_i\in\mathbb R^p$ and $y_i\in\{1,2\}$. Let $n_k$ denote the number of observations in class $k$, $k=1,2$, so that $n=n_1+n_2$, and let $\mathbf{x}_{ki}$ be the $i$th observation from class $k$. Throughout, we work under the balanced equal-prior setting $n_1=n_2$ and $\pi_1=\pi_2=1/2$. 

Classical LDA is motivated by the homoscedastic Gaussian model
\[
\mathbf{x}\mid(y=k)\sim \mathcal{N}(\boldsymbol{\mu}_k,\boldsymbol{\Sigma}),
\qquad k=1,2,
\]
where the class means differ, but the covariance matrix is common. Here and in the classification rules below, $\mathbf{x}\in\mathbb R^p$ denotes a generic point to be classified. If $\boldsymbol{\mu}_1,\boldsymbol{\mu}_2$ and $\boldsymbol{\Sigma}$ are known, the Bayes rule under the equal-prior assumption assigns $\mathbf{x}$ to class $1$ when
\begin{equation}\label{eq:pop_score}
\xi(\mathbf{x})
:=
\left(\mathbf{x}-\frac{\boldsymbol{\mu}_1+\boldsymbol{\mu}_2}{2}\right)^{\top}
\boldsymbol{\Sigma}^{-1}(\boldsymbol{\mu}_1-\boldsymbol{\mu}_2)
>0,
\end{equation}
and to class $2$ otherwise. The population separation is measured by the Mahalanobis distance
\begin{equation}\label{eq:maha}
\Delta
:=
\left[
(\boldsymbol{\mu}_1-\boldsymbol{\mu}_2)^{\top}
\boldsymbol{\Sigma}^{-1}
(\boldsymbol{\mu}_1-\boldsymbol{\mu}_2)
\right]^{1/2}.
\end{equation}
Under the Gaussian model, $\xi(\mathbf{x})\mid(y=k)$ is normal with mean $(-1)^{k+1}\Delta^2/2$ and variance $\Delta^2$. Hence, the Bayes error is $R_{\mathrm{Bayes}}:=\Phi(-\Delta/2)$, where $\Phi$ denotes the standard normal distribution function. Thus, higher values of $\Delta$ correspond to better class separation and a lower Bayes error. In particular, $R_{\mathrm{Bayes}}\to1/2$ as $\Delta\to0$ and $R_{\mathrm{Bayes}}\to0$ as $\Delta\to\infty$.

Given the training observations, we define the class sample means and the pooled within-class sample covariance matrix as
\begin{equation}\label{eq:estimators}
\widehat{\boldsymbol{\mu}}_k
:=
\frac{1}{n_k}\sum_{i=1}^{n_k}\mathbf{x}_{ki},
\qquad
\widehat{\boldsymbol{\Sigma}}
:=
\frac{1}{n}\sum_{k=1}^{2}\sum_{i=1}^{n_k}
(\mathbf{x}_{ki}-\widehat{\boldsymbol{\mu}}_k)
(\mathbf{x}_{ki}-\widehat{\boldsymbol{\mu}}_k)^{\top}.
\end{equation}
We normalize $\widehat{\boldsymbol{\Sigma}}$ by $n$ rather than by $n-2$; the two choices are asymptotically equivalent in the regime $p/n\to\gamma$ considered below. When $\widehat{\boldsymbol{\Sigma}}$ is non-singular, the usual plug-in LDA rule replaces the population quantities in~\eqref{eq:pop_score} by their sample counterparts. To cover both $p<n$ and $p>n$ without introducing an explicit regularization parameter, we use the Moore--Penrose pseudoinverse version of this plug-in rule.

\begin{definition}[Pseudoinverse LDA classifier]\label{Def:pseudoinv}
The pseudoinverse plug-in LDA classifier assigns $\mathbf{x}$ to class $1$ if
$\widehat{\xi}(\mathbf{x})>0$ and to class $2$ otherwise, where
\[
\widehat{\xi}(\mathbf{x})
:=
\left(\mathbf{x}-\frac{\widehat{\boldsymbol{\mu}}_1+\widehat{\boldsymbol{\mu}}_2}{2}\right)^{\top}
\widehat{\boldsymbol{\Sigma}}^{+}
(\widehat{\boldsymbol{\mu}}_1-\widehat{\boldsymbol{\mu}}_2)
=
\widehat{\beta}_0+\widehat{\boldsymbol{\beta}}^{\top}\mathbf{x},
\]
with $\widehat{\boldsymbol{\beta}}:=\widehat{\boldsymbol{\Sigma}}^{+}
(\widehat{\boldsymbol{\mu}}_1-\widehat{\boldsymbol{\mu}}_2)$ and
$\widehat{\beta}_0:=-\frac12
(\widehat{\boldsymbol{\mu}}_1+\widehat{\boldsymbol{\mu}}_2)^{\top}
\widehat{\boldsymbol{\beta}}$. Here $\widehat{\boldsymbol{\Sigma}}^{+}$ is the Moore--Penrose pseudoinverse: if
$\widehat{\boldsymbol{\Sigma}}=\sum_{j=1}^{p}\widehat\lambda_j
\widehat{\mathbf v}_j\widehat{\mathbf v}_j^{\top}$ with
$\widehat\lambda_j\geq0$, then
$\widehat{\boldsymbol{\Sigma}}^{+}
=\sum_{j:\widehat\lambda_j>0}\widehat\lambda_j^{-1}
\widehat{\mathbf v}_j\widehat{\mathbf v}_j^{\top}$. In particular, when $\widehat{\boldsymbol{\Sigma}}$ is non-singular, $\widehat{\boldsymbol{\Sigma}}^{+}=\widehat{\boldsymbol{\Sigma}}^{-1}$ and the
rule reduces to the usual plug-in LDA classifier.
\end{definition}

Since the test point $\mathbf{x}$ from class $k$ satisfies $\mathbf{x}\sim\mathcal{N}(\boldsymbol{\mu}_k,\boldsymbol{\Sigma})$ and is independent of the training sample, the score $\widehat{\xi}(\mathbf{x})$ is, conditioning on the training sample, the score of a test point from class $k$ has a conditional mean $\widehat{\beta}_0+\widehat{\boldsymbol{\beta}}^{\top}\boldsymbol{\mu}_k$ and a conditional variance $\widehat{\boldsymbol{\beta}}^{\top}\boldsymbol{\Sigma}\widehat{\boldsymbol{\beta}}$. Therefore, the class-conditional out-of-sample error is
\[
e_k := \Phi\left((-1)^k\frac{\widehat{\beta}_0+\widehat{\boldsymbol{\beta}}^{\top}\boldsymbol{\mu}_k}{
\left(
\widehat{\boldsymbol{\beta}}^{\top}
\boldsymbol{\Sigma}
\widehat{\boldsymbol{\beta}}
\right)^{1/2}}\right), \qquad k=1,2,
\]
where the denominator is almost surely positive. Under equal class priors, the corresponding conditional out-of-sample misclassification error is
\begin{equation}\label{eq:lda_error}
R_{\mathrm{LDA}}:=\frac{1}{2}(e_1+e_2)
=
\frac12\sum_{k=1}^{2}
\Phi\left(
\frac{
(-1)^k
\left(\boldsymbol{\mu}_k-\frac12(\widehat{\boldsymbol{\mu}}_1+\widehat{\boldsymbol{\mu}}_2)\right)^{\top}
\widehat{\boldsymbol{\Sigma}}^{+}
(\widehat{\boldsymbol{\mu}}_1-\widehat{\boldsymbol{\mu}}_2)
}{
\left[
(\widehat{\boldsymbol{\mu}}_1-\widehat{\boldsymbol{\mu}}_2)^{\top}
\widehat{\boldsymbol{\Sigma}}^{+}
\boldsymbol{\Sigma}
\widehat{\boldsymbol{\Sigma}}^{+}
(\widehat{\boldsymbol{\mu}}_1-\widehat{\boldsymbol{\mu}}_2)
\right]^{1/2}
}
\right).
\end{equation}
This exact conditional risk expression~\eqref{eq:lda_error} is the starting point for the high-dimensional analysis.

\section{Main Results}\label{sec:results}

Building on the closed-form error expression~\eqref{eq:lda_error}, we study the high-dimensional regime in which $p$ and $n$ diverge proportionally. We work throughout under the following two assumptions.
\begin{assumption}\label{ASSU:ratio}
As $n,p\to\infty$, we consider $\gamma_n:= p/n\to\gamma\in(0,1)\cup(1,\infty)$.
\end{assumption}

\begin{assumption}\label{ASSU:DGP}
For $k=1,2$, the observations $\mathbf{x}_{ki}$, $i=1,\ldots,n_k$, are independent across $(k,i)$ with $\mathbf{x}_{ki}\sim\mathcal{N}(\boldsymbol{\mu}_k,\boldsymbol{\Sigma})$, where $\boldsymbol{\Sigma}$ is deterministic and positive definite.
\end{assumption}

\begin{assumption}\label{ASSU:signal}
The Mahalanobis distance $\Delta$ in~\eqref{eq:maha} is bounded away from $0$ and $\infty$ as $n,p\to\infty$.
\end{assumption}

Assumption~\ref{ASSU:ratio} is the standard proportional-growth regime in random matrix theory, where the estimation error of $\widehat{\boldsymbol{\Sigma}}$ remains non-negligible asymptotically. Assumption~\ref{ASSU:DGP} is the Gaussian data-generating model underlying classical LDA; it makes the discriminant score exactly Gaussian given the training sample. Assumption~\ref{ASSU:signal} keeps the population classification problem non-degenerate: the Bayes error neither tends to $1/2$ nor to $0$. This keeps the population classification problem at a stable level of difficulty. Hence, any limiting loss in performance is attributable to the high-dimensional estimation of the LDA rule, rather than to a vanishing or diverging population separation.

% \begin{assumption}\label{ASSU:DGP}
% For $k=1,2$ and $i=1,\ldots,n_k$, the observations satisfy $\mathbf{x}_{ki}=\boldsymbol{\mu}_k+\boldsymbol{\Sigma}^{1/2}\mathbf{z}_{ki}$, where $\boldsymbol{\Sigma}$ is deterministic and positive definite, and the $\mathbf{z}_{ki}$ are independent across $(k,i)$ with i.i.d.\ standard Gaussian entries, i.e.\ $\mathbf{z}_{ki}\stackrel{\textup{i.i.d.}}{\sim}\mathcal{N}(\mathbf{0},\mathbf{I}_p)$, so that $\mathbf{x}_{ki}\sim\mathcal{N}(\boldsymbol{\mu}_k,\boldsymbol{\Sigma})$.
% \end{assumption}

\subsection{Under-parameterized Regime}\label{sec:under}

When $\gamma<1$, the under-parameterized regime, $\widehat{\boldsymbol{\Sigma}}$ is invertible, and $\widehat{\boldsymbol{\Sigma}}^{+}=\widehat{\boldsymbol{\Sigma}}^{-1}$. Our first result gives the limiting error in this regime for general $\boldsymbol{\Sigma}$. Notably, no structural conditions are imposed on $\boldsymbol{\Sigma}$ beyond positive definiteness: spiked or diverging eigenvalues are allowed.

\begin{theorem}
\label{Thm1}
Under Assumption~\ref{ASSU:ratio}-\ref{ASSU:signal} with $p<n$, the out-of-sample error $R_{\mathrm{LDA}}$ in~\eqref{eq:lda_error} of the classifier of Definition~\ref{Def:pseudoinv} satisfies
\begin{align*}
R_{\mathrm{LDA}} - \Phi \left( - \frac{\sqrt{1-\gamma}\Delta^{2}}{2\sqrt{\Delta^{2}+4\gamma}} \right)\convergeInP0,
\end{align*}
where ``$\text{i.p.}$'' denotes convergence in probability as $n,p\to\infty$.
\end{theorem}
Theorem~\ref{Thm1} shows that the limiting error depends on the data model only through the Mahalanobis distance $\Delta$ and the dimensional ratio $\gamma$. The error decreases in $\Delta$: as $\Delta$ increases, the argument of $\Phi$ becomes more negative, so the error falls from random guessing toward~$0$.

The limit can be written as
\[
\Phi\left(-\frac{\Delta}{2}\cdot \frac{1}{\sqrt{T_{\mu}T_{\Sigma}}}\right),
\qquad
T_{\mu}:=1+\frac{4\gamma}{\Delta^{2}},
\qquad
T_{\Sigma}:=\frac{1}{1-\gamma},
\]
where $T_{\mu}$ captures the effect of mean estimation and $T_{\Sigma}$ captures the effect of covariance estimation. In particular, if the population means are known, then $T_{\mu}\equiv 1$ and the limit reduces to $\Phi(-\Delta\sqrt{1-\gamma}/2)$.

To interpret the limiting expression, first maintain the population separation $\Delta$ fixed at some value bounded away from $0$ and $\infty$ and vary the dimensional ratio $\gamma$. As $\gamma\to0$, both $T_{\mu}\to1$ and $T_{\Sigma}\to1$, so the limiting error recovers the Bayes error $\Phi(-\Delta/2)$, as expected in the classical low-dimensional regime. As $\gamma\uparrow1$, the covariance-estimation penalty $T_{\Sigma}$ diverges and the limiting error approaches $\Phi(0)=1/2$, reflecting the ill-conditioning of the sample covariance matrix near $\gamma=1$. This boundary behavior is consistent with Theorem~2.2 of \citet{Wang2018} for the case of known class means. For a given $\Delta$, the limiting error increases monotonically with $\gamma$ in the under-parameterized regime, as both penalties $T_{\mu}$ and $T_{\Sigma}$ grow. This describes the effect of dimensionality with the signal held fixed; Assumption~\ref{ASSU:signal} only bounds $\Delta$ and does not fix it across dimensions, so it remains compatible with the growing-signal simulations of Section~\ref{sec:empirics}.

The U-shaped curves observed in the simulations arise because the population separation may change as features are added. To make the dimension dependence of the mean difference explicit, define
\begin{equation}\label{eq:mean_difference}
\boldsymbol{\delta}_p:=\boldsymbol{\mu}_1-\boldsymbol{\mu}_2.
\end{equation}
In the simulations, we set $\boldsymbol{\delta}_p=\mathbf 1_p$. Consequently, the Mahalanobis distance $\Delta$ defined in~\eqref{eq:maha} may vary with $p$, with
$\Delta^2=\boldsymbol{\delta}_p^{\top}\boldsymbol{\Sigma}^{-1}\boldsymbol{\delta}_p$.
Its behavior depends on the covariance structure. Under identity and AR(1) covariance with fixed parameters, $\Delta^2$ grows with $p$ and eventually falls outside the bounded-signal regime of Assumption~\ref{ASSU:signal}. Under compound-symmetry covariance with fixed $\rho>0$, by contrast, $\Delta^2=p/\{\sigma^2(1-\rho+\rho p)\}$ remains bounded as $p\to\infty$.

Increasing the dimension can therefore have two competing effects. Additional coordinates may strengthen population separation, while the larger dimensional ratio increases covariance-estimation error. At relatively small values of $\gamma$, the improvement in separation can dominate and reduce error. As $\gamma$ approaches one, the divergence of $T_{\Sigma}$ causes the estimation effect to dominate and the error to rise toward its peak. This competition explains the U-shaped error curves shown in Figures~\ref{fig:4.1} and~\ref{fig:2}, including the local minimum at an intermediate value of $\gamma$ also observed by \citet{raudys1998}.
\begin{figure}[!htb]
\centering
\includegraphics[scale = 0.3]{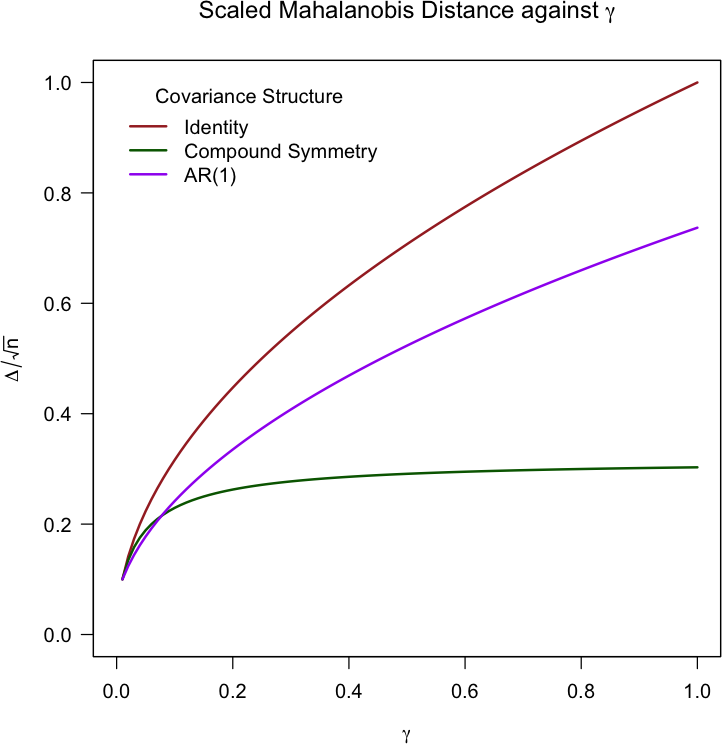}
\caption{Scaled Mahalanobis distance $\Delta/\sqrt{n}$ as a function of $\gamma$ under identity, AR(1), and compound-symmetry covariance structures. We set $\boldsymbol{\delta}_p=\boldsymbol{\mu}_1-\boldsymbol{\mu}_2=\mathbf{1}_p$ and $n=100$, with $\rho=0.3$ for the AR(1) covariance and $\rho=0.1$ for the compound-symmetry covariance.}\label{fig:4.1}
\end{figure}

\begin{figure}[!htb]
\centering
\begin{subfigure}{0.45\textwidth}
\includegraphics[width=\textwidth]{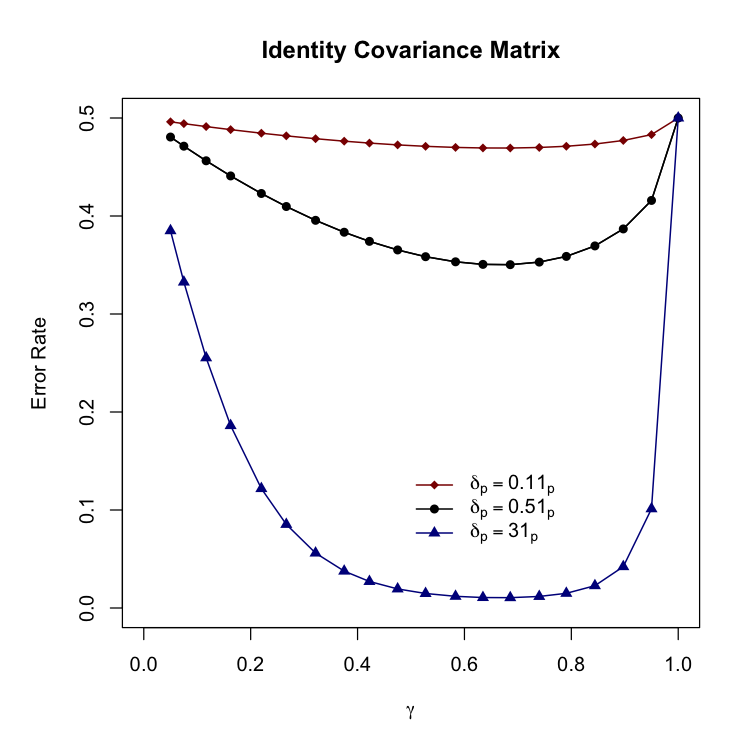}
\caption{Identity covariance with
$\boldsymbol{\delta}_p=c\mathbf 1_p$ for
$c\in\{0.1,0.5,3\}$.}
\label{fig:identity_error}
\end{subfigure}
\hfill
\begin{subfigure}{0.45\textwidth}
\includegraphics[width=\textwidth]{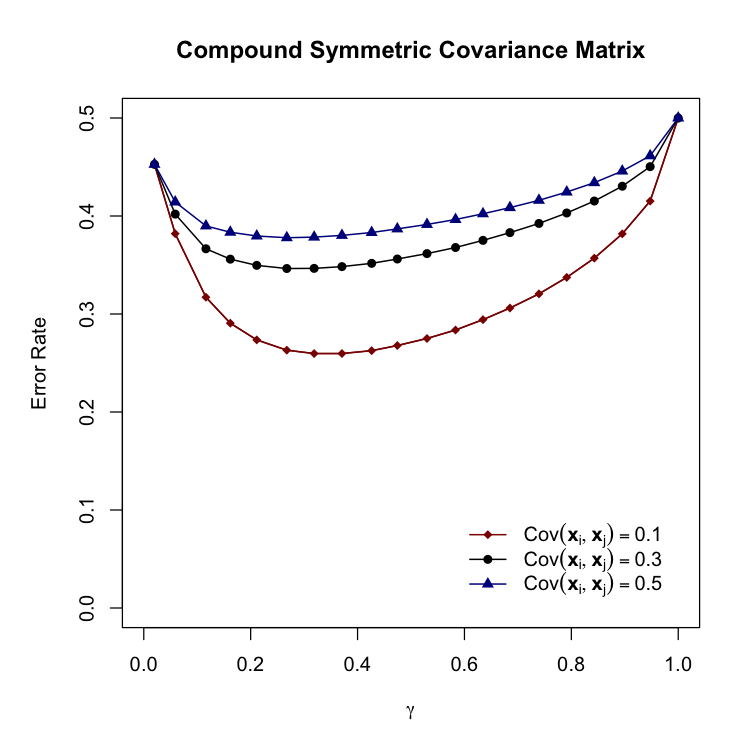}
\caption{Compound-symmetry covariance with
$\boldsymbol{\delta}_p=3\mathbf 1_p$ and
$\rho\in\{0.1,0.3,0.5\}$.}
\label{fig:compound_error}
\end{subfigure}
\caption{Under-parameterized LDA error curves illustrating the effects of population mean separation and feature dependence. The sample size varies over $n\in[20,400]$, with $p/n\in(0,1)$.}
\label{fig:2}
\end{figure}

\subsection{Over-parameterized Regime}\label{sec:over}

When $\gamma>1$, the pooled sample covariance matrix $\widehat{\boldsymbol{\Sigma}}$ is singular and the classifier of Definition~\ref{Def:pseudoinv} uses the Moore--Penrose pseudoinverse $\widehat{\boldsymbol{\Sigma}}^{+}$. We focus on the isotropic covariance setting, which yields an explicit characterization of the limiting error; the general covariance case is discussed at the end of this subsection.

\begin{theorem}\label{Thm2}
Under Assumption~\ref{ASSU:ratio}-\ref{ASSU:signal} with $p>n$, suppose that $\boldsymbol{\Sigma}=\sigma^2\mathbf I_p$ for some fixed $\sigma^2>0$. Then the out-of-sample error $R_{\mathrm{LDA}}$ in~\eqref{eq:lda_error} of the classifier of Definition~\ref{Def:pseudoinv} satisfies
\[
R_{\mathrm{LDA}}-\Phi\!\left(-\frac{\sqrt{\gamma-1}\,\Delta^{2}}{2\gamma\,\sqrt{\Delta^{2}+4\gamma}}\right)\convergeInP0.
\]
\end{theorem}
Paralleling Theorem~\ref{Thm1}, the limit factorizes as
\[
\Phi\!\left(
-\frac{\Delta}{2\sqrt{T_{\mu}^{+}T_{\Sigma}^{+}}}
\right),
\qquad
T_{\mu}^{+}:=1+\frac{4\gamma}{\Delta^2},
\qquad
T_{\Sigma}^{+}:=\frac{\gamma^2}{\gamma-1}.
\]
Since Theorem~\ref{Thm1} includes the isotropic case, the two regimes can be compared under $\boldsymbol{\Sigma}=\sigma^2\mathbf I_p$. The mean-estimation term has the same functional form, $1+4\gamma/\Delta^2$, on both sides of $\gamma=1$, whereas the covariance-estimation term is $(1-\gamma)^{-1}$ for $\gamma<1$ and $\gamma^2/(\gamma-1)$ for $\gamma>1$.

As $\gamma$ approaches one from either side, the covariance-estimation term diverges and the limiting error approaches $\Phi(0)=1/2$. The two expressions therefore are continuously joined at $\gamma=1$. Spectrally, the lower edge of the relevant sample-covariance spectrum approaches zero from both sides, making either the inverse or the pseudoinverse unstable near the threshold.

Immediately after $\gamma=1$, the limiting error decreases from $1/2$. For a given $\Delta$, this behavior is determined by ${\sqrt{\gamma-1}}/(
{\gamma\sqrt{\Delta^2+4\gamma}})$, which initially increases as $\gamma$ moves into the over-parameterized regime. The corresponding argument of $\Phi$ becomes more negative, producing the second descent. An interpretation is that the nonzero eigenvalues of the sample covariance matrix move away from zero as $\gamma$ increases, reducing the instability of the pseudoinverse.

The second descent does not continue indefinitely when $\Delta$ remains bounded. The same factor eventually decreases to zero as in $\gamma\to\infty$. Hence, the limiting error reaches a minimum at an intermediate value of $\gamma>1$ and then gradually returns to $1/2$. Under extreme over-parameterization, the pseudoinverse operates only on the sample span, whose dimension is at most $n$, while the ambient dimension continues to grow. With bounded population separation, the effective discriminatory information retained by the classifier eventually becomes insufficient to maintain the improvement obtained during the second descent.

Figure~\ref{fig:error_3D} combines the limits in Theorems~\ref{Thm1} and~\ref{Thm2}. For each given value of $\Delta$, the error increases toward a ridge of height $1/2$ as $\gamma\uparrow1$, then decreases on the over-parameterized side, forming the second-descent region. The surface subsequently turns upward as $\gamma$ becomes large, reflecting the return of the limiting error toward $1/2$. Increasing $\Delta$ lowers the surface for every $\gamma\neq1$ and deepens the over-parameterized valley, since stronger population separation improves classification in both regimes. However, in $\gamma=1$, the limiting error remains $1/2$ for all $\Delta$, because the covariance-estimation instability dominates the population signal.

Thus, for a given bounded $\Delta$, the theoretical surface shows an ascent toward random guessing as $\gamma$ approaches one, a second descent after $\gamma$ exceeds one, and a later deterioration under extreme over-parameterization. A sustained decline toward zero requires the population separation itself to increase with dimension, as examined separately in Section~\ref{sec:empirics}.

\begin{figure}[!htb]
\centering
\includegraphics[width=0.48\textwidth]{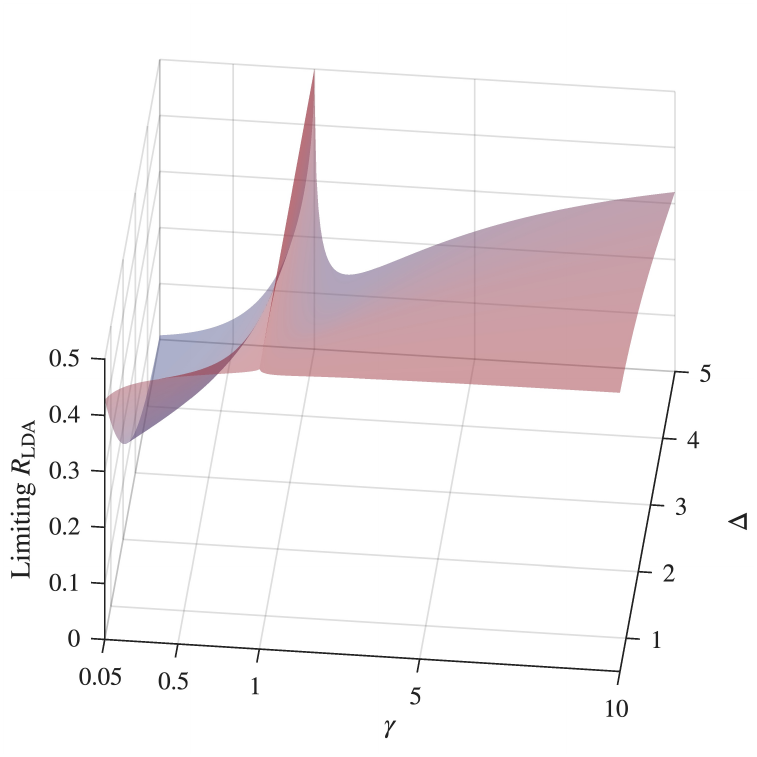}
\caption{Limiting out-of-sample misclassification error of the pseudoinverse LDA classifier as a function of the dimensional ratio $\gamma$ and the Mahalanobis distance $\Delta$. The surface for $\gamma<1$ is given by Theorem~\ref{Thm1}, while the surface for $\gamma>1$ is given by the isotropic result in Theorem~\ref{Thm2}. The two expressions have the common one-sided limit $1/2$ at $\gamma=1$ and therefore join continuously at the threshold.}
\label{fig:error_3D}
\end{figure}

% \subsection{Connection to Benign Overfitting}

% We recall that benign overfitting refers to the phenomenon in which increasing features to close to or more than the sample size leads to an improvement in the error rate \citep{Bartlett2020}. \rev{Theorems~1 and~2 together provide a theoretical account of this phenomenon for LDA in the isotropic model: the risk curve is continuous across the interpolation threshold $\gamma=1$, where it peaks at random guessing, and---when the added features are informative so that $\Delta^2$ grows with the dimension---descends a second time in the over-parameterized regime $\gamma>1$.} We also notice that in the double asymptotic setting, the relationship between LDA error and model flexibility is largely influenced by the dimension effect of estimating the population parameters. \rev{For a general covariance structure, where only the under-parameterized result is proven, we} further investigate this phenomenon using simulated data\rev{, with a focus on the $\gamma>1$ regime, in the next section.}

Theorems~\ref{Thm1} and~\ref{Thm2} characterize the limiting error under the covariance settings covered by the theory. In the next section, we use simulations to examine whether the same double-descent pattern persists under more general covariance structures and model settings.

\section{Simulation Analysis}\label{sec:empirics}

This section serves two purposes. First, we examine the finite-sample accuracy of the asymptotic results in Theorems~\ref{Thm1} and~\ref{Thm2} under the isotropic Gaussian setting covered by the theory. Second, we investigate whether the same double-descent patterns persist under more general covariance structures, non-Gaussian feature distributions, and different forms of label and feature noise.

\subsection{Simulation Design}\label{sec:sim-design}

We consider $20$ simulation designs indexed by $j=1,2,\ldots,20$. In design $j$, the two classes are balanced, with $n_{1j}=n_{2j}$ observations each, and the dimension is $p_j$. The sample size per-class ranges from $20$ to $50$ and $p_j$ ranges from $5$ to $400$. Within each class, the observations are randomly split into $70\%$ training and $30\%$ tests; writing $n_j$ for the total training size, we set $\gamma_j:=p_j/n_j$ and retain the designs with $\gamma_j\in[0.5,4]$.

For each design, the population parameters $\boldsymbol{\mu}_1$, $\boldsymbol{\mu}_2$, and $\boldsymbol{\Sigma}$ are chosen according to the covariance and signal structure under consideration. The raw mean difference $\boldsymbol{\delta}_p=\boldsymbol{\mu}_1-\boldsymbol{\mu}_2$ specifies the design, while the resulting separation is measured by the Mahalanobis distance $\Delta$ of~\eqref{eq:maha}, with $\Delta^2=\boldsymbol{\delta}_p^{\top}\boldsymbol{\Sigma}^{-1}\boldsymbol{\delta}_p$; the same raw mean difference may correspond to different classification difficulties under different covariance structures.

Under the baseline (Gaussian) design, observations are drawn independently as $\mathbf{x}_{ki}\sim\mathcal{N}(\boldsymbol{\mu}_k,\boldsymbol{\Sigma})$ for $k\in\{1,2\}$. A pseudo-inverse LDA classifier (Definition~\ref{Def:pseudoinv}) is fitted to the training data and evaluated on the test set, giving the empirical misclassification error
\[
\widehat R
=
\frac{1}{n_j^{\mathrm{te}}}
\sum_{i=1}^{n_j^{\mathrm{te}}}
\mathbf 1\bigl(y_i\neq\widehat y_i\bigr),
\]
where $n_j^{\mathrm{te}}$ is the number of test observations, $y_i$ the true label, and $\widehat y_i$ the predicted label. Each design is repeated $100$ times and the reported error is averaged over replications.

For Gaussian designs, we also evaluate the exact conditional error~\eqref{eq:lda_error} from the training-sample estimates and the population parameters. Because this expression integrates over an independent Gaussian test point, it avoids the additional Monte Carlo variation of a finite test set. We compare the average conditional error with the empirical test error and, where applicable, with the asymptotic limits of Theorems~\ref{Thm1} and~\ref{Thm2}. For designs outside the Gaussian model, performance is assessed by empirical test error alone.

The remaining subsections report results by experimental factor: sample size (Section~\ref{sec:sim-sample}), covariance structure (Section~\ref{sec:sim-cov}), non-Gaussian features (Section~\ref{sec:non-normal}), label and feature noise (Section~\ref{sec:noise}) and redundant features (Section~\ref{sec:redundant}).

\subsection{Effect of Sample Size}\label{sec:sim-sample}

In \figref{fig:sample_size}, we select three sample sizes to visualize the error curve with varying numbers of features. The error curves exhibit double descent, showing a U-shaped in the under-parameterized regime, peaks at the interpolation threshold $p \simeq n$, and decreases afterward. We notice an interesting phenomenon that at different values of $p$, more training samples do not lead to better performance because increasing the sample size shifts the interpolation threshold to the right, leading to a higher error rate at a fixed $p$. This observation is consistent with the behavior observed by \cite{deng2021} for logistic regression with polynomial features and \citep{Nakkiran2021(biggermodels)} in their analysis of the neural network model performance. 
\begin{figure}[!htb]
\centering
\includegraphics[scale = 0.45]{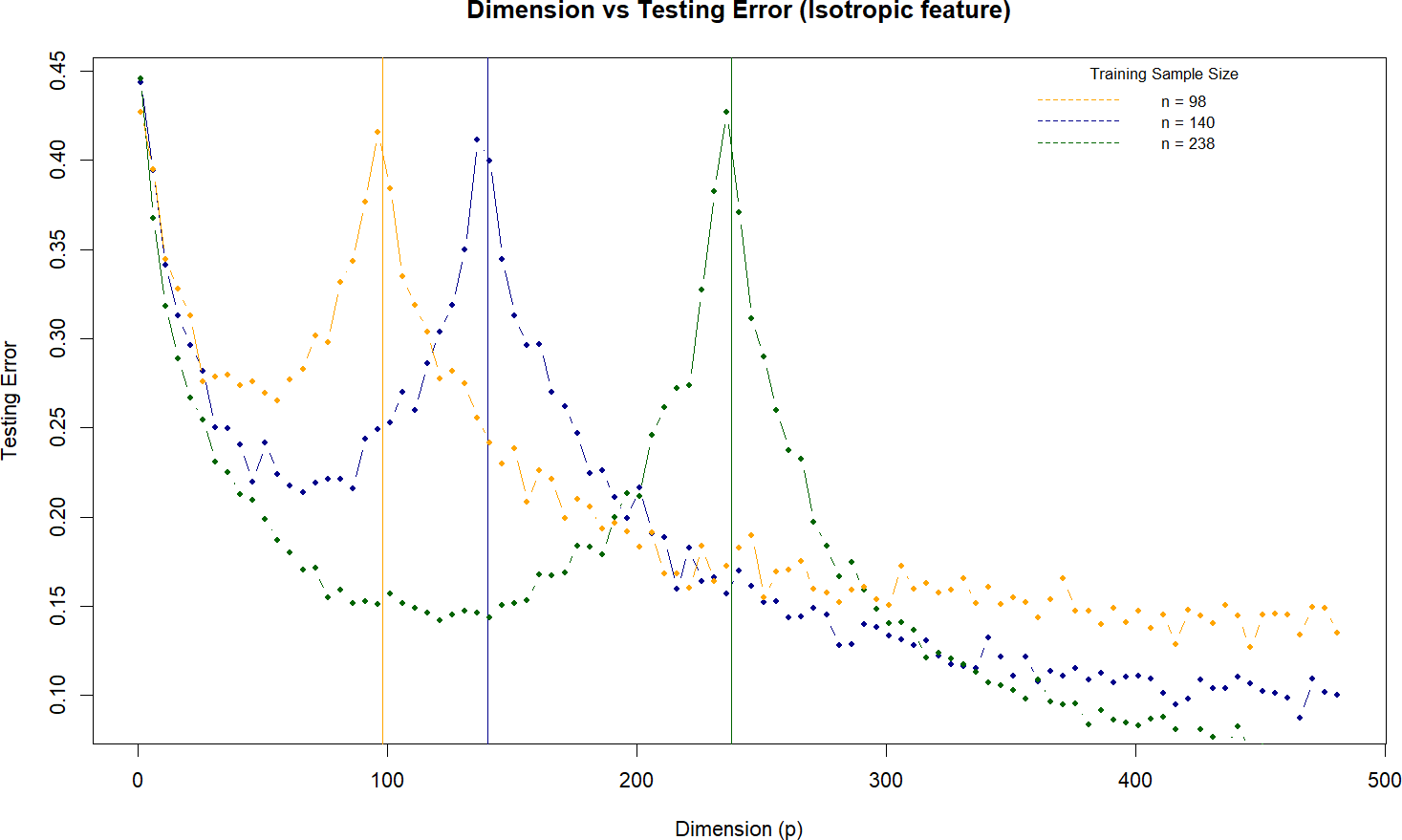}
\caption{Simulated error rate with isotropic features for fixed $n$ and increasing $p$ ($\boldsymbol{\delta}_p = 0.1\mathbf 1_p$, $\boldsymbol{\Sigma} = \mathbf{I}_p$).}
\label{fig:sample_size}
\end{figure}

\subsection{Effect of Covariance Structure}\label{sec:sim-cov}

\noindent \textbf{1. Isotropic Covariance Structure $\boldsymbol{\Sigma}=\mathbf{I}_p$.} In~\figref{fig:iso_cov}, we increase both $n$ and $p$. The test error becomes larger as the magnitude of the population mean difference decreases, reflecting the increasing difficulty of separating the two classes. Specifically, we set $\boldsymbol{\delta}_p=c\mathbf{1}_p$ for $c\in\{0.5,1,1.5\}$. For each signal level, the error increases as $\gamma$ approaches one and then decreases sharply in the over-parameterized regime. Since the population covariance is isotropic, there is no heterogeneity across the covariance directions. As $\gamma$ moves away from the one in the over-parameterized regime, the pseudoinverse becomes better conditioned, contributing to the rapid decline in error.

Theorem~\ref{Thm2} explains the initial second descent immediately beyond $\gamma=1$. However, in the present simulation, $\Delta^2=c^2p$ also increases with dimension. This growing-signal design eventually lies outside the bounded-signal condition in Assumption~\ref{ASSU:signal} and further reinforces the decline. This explains why the error continues toward zero over the range of $\gamma$ considered rather than eventually returning to $1/2$, as in the bounded-signal limit.
\begin{figure}[!htb]
\centering
\includegraphics[scale=0.45]{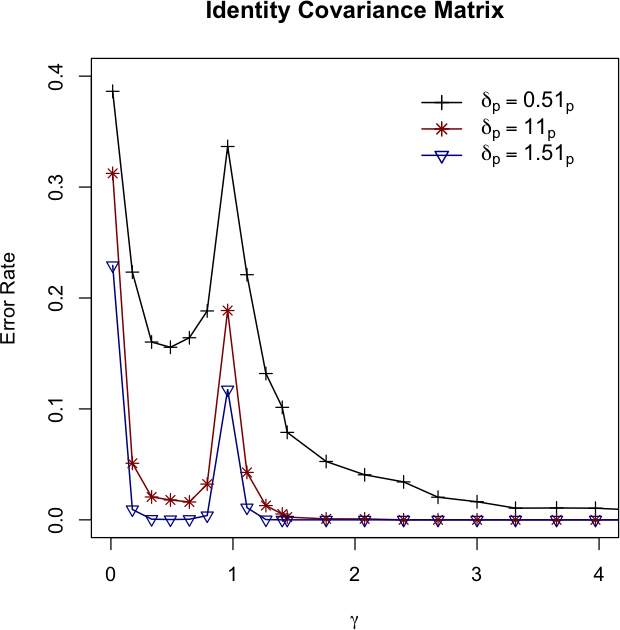}
\caption{Average test error under identity covariance with increasing $n$, $p$, and $\gamma$, where $\boldsymbol{\delta}_p=c\mathbf{1}_p$ for $c\in\{0.5,1,1.5\}$.}
\label{fig:iso_cov}
\end{figure}

\noindent \textbf{2. Compound Symmetry Structure.}
We introduce a common dependence among the features by setting the off-diagonal entries of $\boldsymbol{\Sigma}$ equal to $\rho$:
\[
\boldsymbol{\Sigma} = \sigma^2
\begin{bmatrix}
1 & \rho & \rho & \rho \\
\rho & 1 & \rho & \rho \\
\vdots & \vdots & \vdots & \vdots \\
\rho & \rho & \rho & 1
\end{bmatrix}.
\]

The simulation and conditional-error results show that the error increases with the covariance scale $\sigma^2$, which reduces the Mahalanobis separation; see~\figref{fig:comp_sym}. In contrast to the identity case, the initial decrease in error is weaker. This is due to the slower increase in the Mahalanobis distance under the dependence of common features. When $\boldsymbol{\delta}_p=\mathbf{1}_p$,
$
\Delta^2
=
\frac{p}{\sigma^2(1-\rho+\rho p)},
$
which converges to $1/(\sigma^2\rho)$ for $\rho>0$. Thus, the additional features become increasingly redundant and the population separation gradually levels off.

For $\gamma>1$, the error decreases more slowly than under identity covariance, and the minimum occurs at a larger value of $\gamma$ in the designs considered. The signal direction $\mathbf{1}_p$ is the leading eigenvector of the compound-symmetry covariance matrix, with eigenvalue $\sigma^2(1-\rho+\rho p)$. Its linear growth with $p$ limits the effective discriminatory information contributed by additional coordinates and slows the reduction in error. Since Theorem~\ref{Thm2} is proved only under isotropic covariance, the right panel of Figure~\ref{fig:comp_sym} reports the exact Gaussian conditional error rather than the asymptotic limit in Theorem~\ref{Thm2}.
\begin{figure}[!htb]
\centering
\begin{subfigure}{0.48\textwidth}
\includegraphics[width=\textwidth]{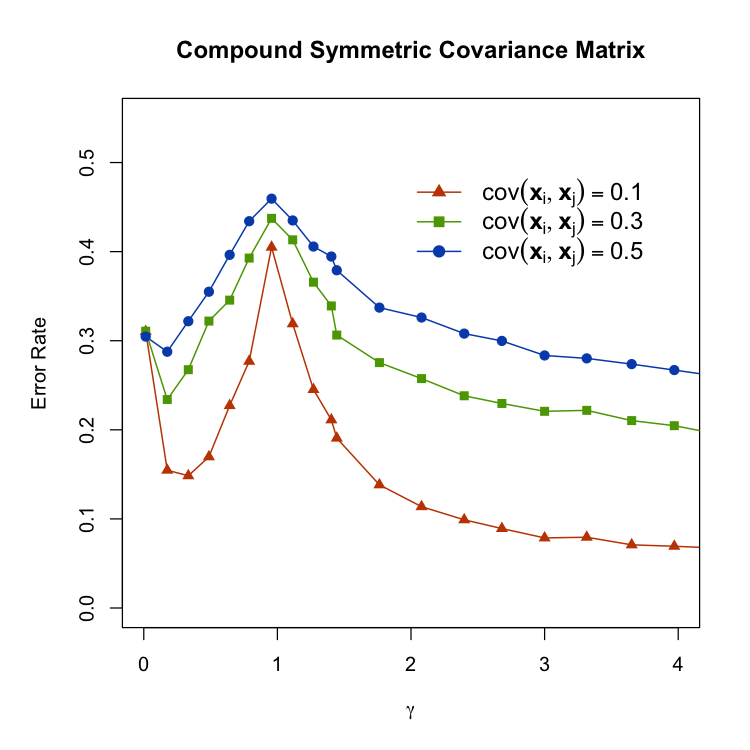}
\caption{Average test error.}
\label{fig:5.3(a)}
\end{subfigure}
\hfill
\begin{subfigure}{0.48\textwidth}
\includegraphics[width=\textwidth]{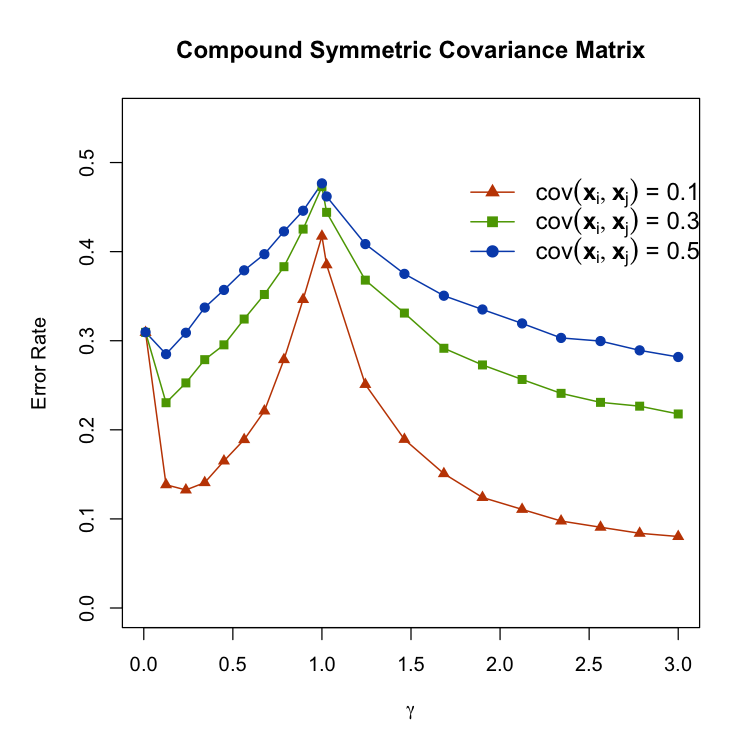}
\caption{Average conditional error.}
\label{fig:5.3(b)}
\end{subfigure}
\caption{Test and conditional errors under compound-symmetry covariance with $\boldsymbol{\delta}_p=\mathbf{1}_p$.}
\label{fig:comp_sym}
\end{figure}

\medskip
\noindent \textbf{3. First-Order Autoregressive Structure, AR(1).}
Under the AR(1) structure, the correlation between two features decreases exponentially with their distance:
\begin{flalign*}
\boldsymbol{\Sigma} &= \sigma^2
\begin{bmatrix}
1 & \rho & \rho^2 & \dots & \rho^{p-1} \\
\rho & 1 & \rho & \dots & \rho^{p-2} \\
\rho^2 & \rho & 1 & \ddots & \vdots \\
\vdots & \vdots & \ddots & \ddots & \rho \\
\rho^{p-1} & \rho^{p-2} & \dots & \rho & 1
\end{bmatrix}.
\end{flalign*}

Figure~\ref{fig:ar1} presents the average test and conditional errors. The AR(1) covariance structure behaves similarly to compound symmetry near $\gamma=1$, with the error increasing toward the threshold and decreasing in the over-parameterized regime. However, the error curves for different values of $\rho$ become closer at large $\gamma$. Unlike compound symmetry, the dependence of AR$(1)$ is local, and its population spectrum remains bounded for a fixed $|\rho|<1$.

For $\boldsymbol{\delta}_p=\mathbf{1}_p$, the squared Mahalanobis distance is
\[
\Delta^2 = \frac{p(1-\rho)+2\rho}{\sigma^2(1+\rho)},
\]
and therefore grows linearly with $p$. A larger $\rho$ reduces its growth rate, but additional coordinates continue to contribute discriminatory information. The sustained decline in the over-parameterized regime is therefore driven by both the improved conditioning of the pseudoinverse away from $\gamma=1$ and the increasing population separation. As in the identity case, this growing-signal behavior lies beyond the formal scope of Theorem~\ref{Thm2}.

\begin{figure}[!htb]
\centering
\begin{subfigure}{0.47\textwidth}
\includegraphics[width=\textwidth]{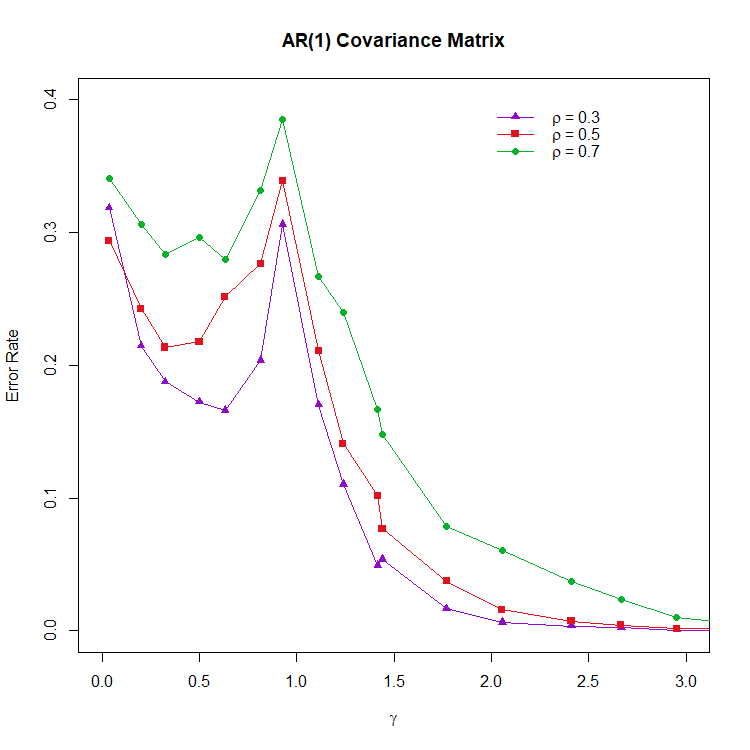}
\caption{Average test error.}
\label{fig:5.3(a_2)}
\end{subfigure}
\hfill
\begin{subfigure}{0.47\textwidth}
\includegraphics[width=\textwidth]{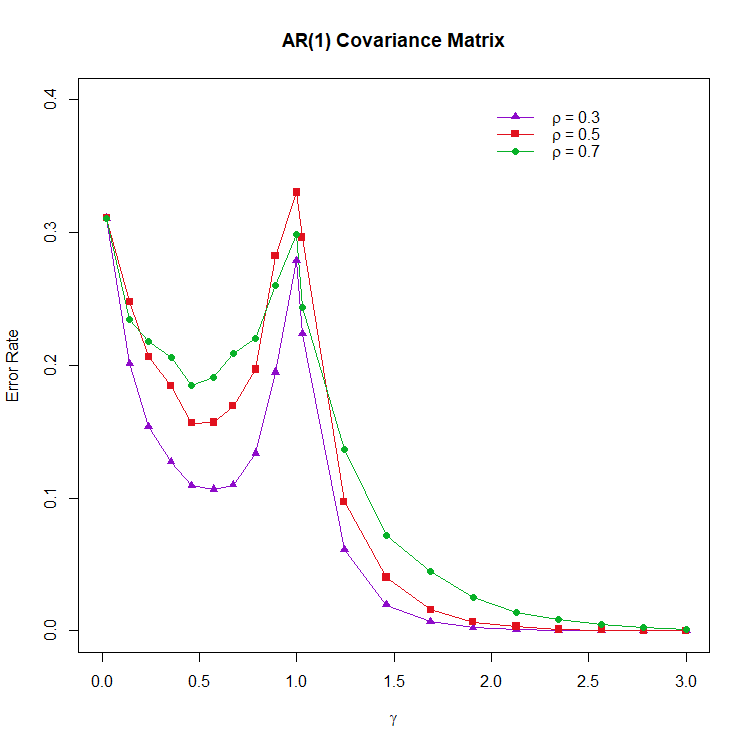}
\caption{Average conditional error.}
\label{fig:5.4}
\end{subfigure}
\caption{Test and conditional errors under AR(1) covariance for different values of $\rho$, with $\boldsymbol{\delta}_p=\mathbf{1}_p$.}
\label{fig:ar1}
\end{figure}

\subsection{Non-Gaussian Data}\label{sec:non-normal}

The true distribution of the data is often unknown in reality, and LDA can be applied to datasets that follow a non-normal distribution. Since our theoretical analysis allows individual feature vectors $\mathbf{x}_i$ to be non-Gaussian, we investigate the double descent phenomenon when the normality assumption is violated.

We first consider the case where the features follow an independently uniform and Poisson distribution (\figref{fig:unif_poiss}). We observe that LDA behaves similarly to the identity covariance matrix case. The test error peaks around $p \approx n$ and tails off asymptotically towards 0 for large $\gamma$.
\begin{figure}[!htb]
\centering
\begin{subfigure}{0.48\textwidth}
\includegraphics[width=\textwidth]{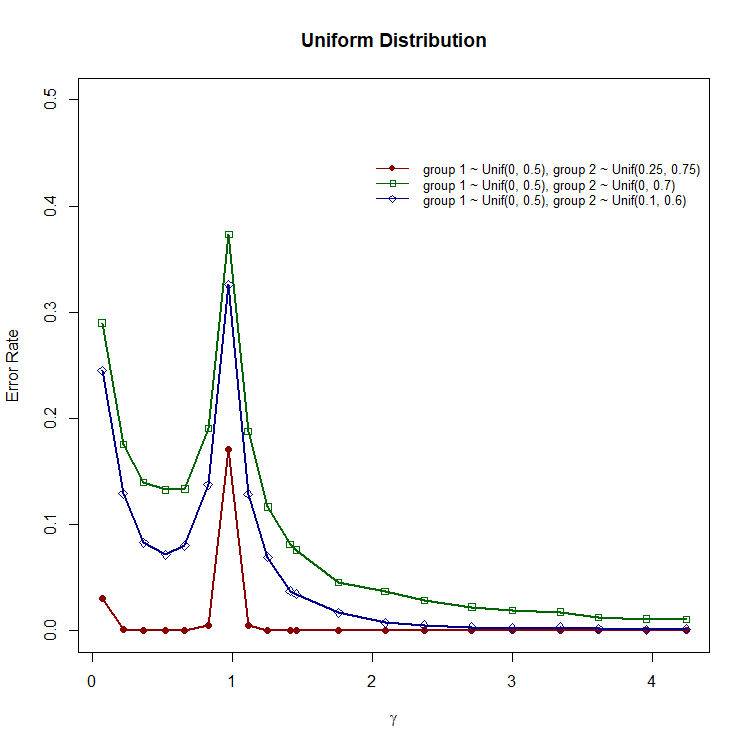}
\caption{$\mathbf{x}_{ki} \sim U[\min_k,\max_k]$, $k \in \{1,2\}$\label{fig:5.6(a)}}
\end{subfigure}
\hfill
\begin{subfigure}{0.48\textwidth}
\includegraphics[width=\textwidth]{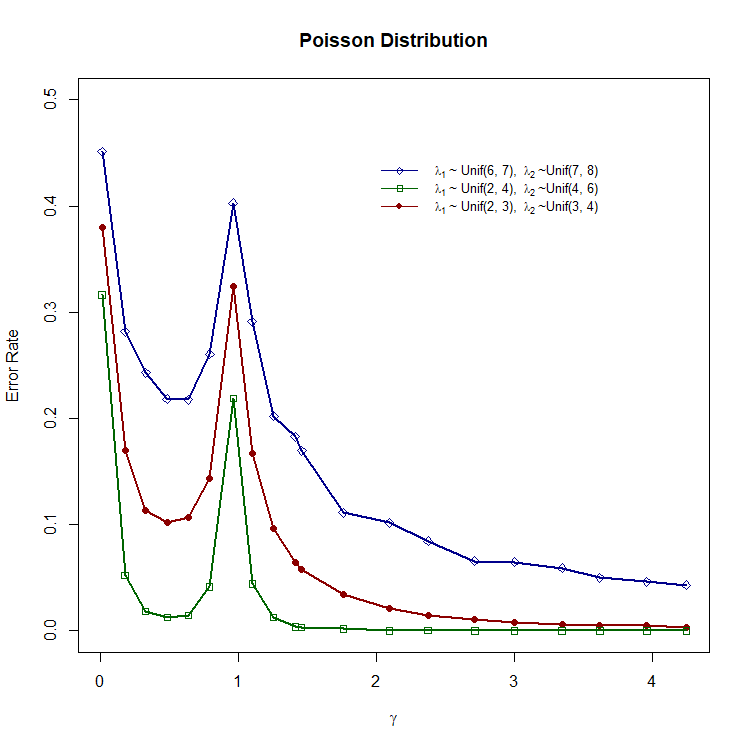}
\caption{$\mathbf{x}_{ki} \sim \text{Pois}(\lambda_k)$, $k \in \{1,2\}$}
\label{fig:5.6(b)}
\end{subfigure}
\caption{Independently distributed Features}
\label{fig:unif_poiss}
\end{figure}

The $T$-distribution has a heavier tail than the normal distribution and {violates the Gaussian assumption underlying our theory.}  However, the simulation shows that the error curve performs similarly to the normal samples, and the heavier tails contribute to larger error (\figref{fig:T-dis}).
\begin{figure}[!htb]
\centering
\begin{subfigure}{0.45\textwidth}
\includegraphics[width=\textwidth]{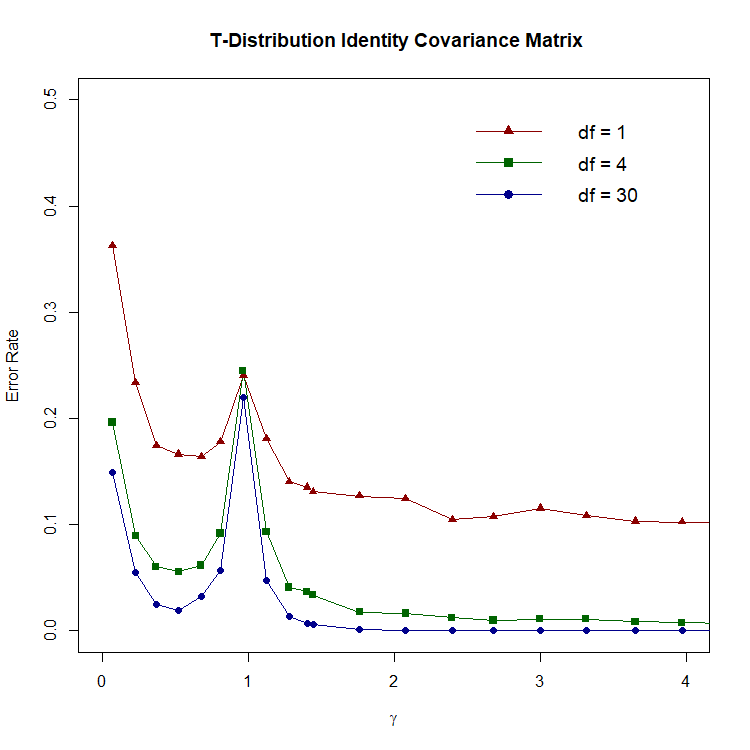}
\caption{T-distribution with identity covariance matrix}
\label{fig:5.7(a)}
\end{subfigure}
\hfill
\begin{subfigure}{0.45\textwidth}
\includegraphics[width=\textwidth]{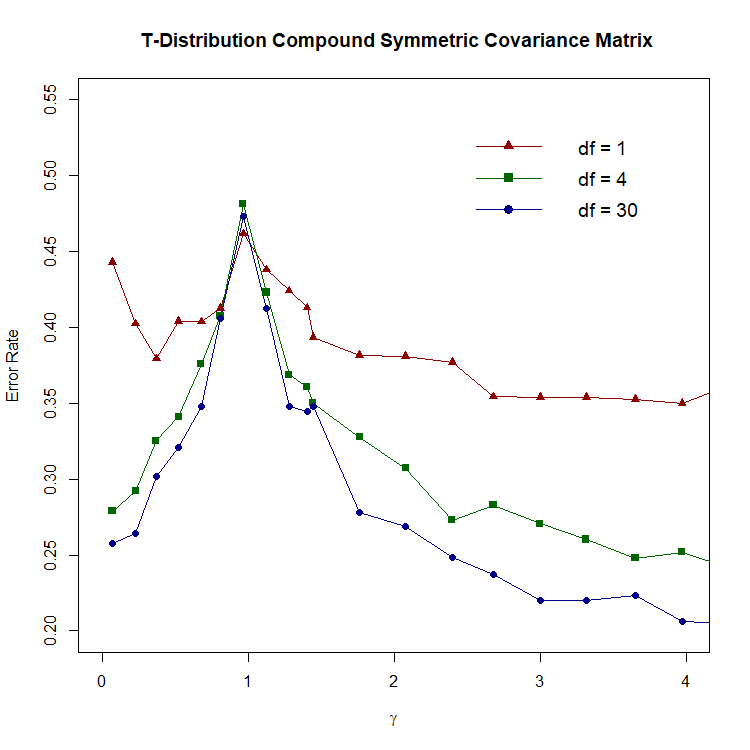}
\caption{T-distribution with compound symmetric covariance matrix ($\rho = 0.3$)}\label{fig:5.7(b)}
\end{subfigure}
\caption{$T$-distributed data}\label{fig:T-dis}
\end{figure}

For logarithmic-normal and Dirichlet features, the test error shows no systematic dependence on $\gamma$ and remains close to random guessing, with no clear double-descent pattern. This result should be interpreted as evidence of model mis-specification rather than as an intrinsic failure of high-dimensional LDA. Both distributions depart substantially from the homoscedastic Gaussian model underlying the theoretical analysis. Log-normal features are strongly right-skewed, while Dirichlet features are constrained to the simplex and exhibit dependence induced by the sum-to-one constraint. Consequently, a linear discriminant rule based only on the class means and a common covariance matrix may be poorly aligned with the Bayes classifier. Changes in the dimensional ratio alone, therefore, do not need to generate the double-descent mechanism described by Theorems~\ref{Thm1} and~\ref{Thm2}. These experiments assess the robustness of LDA beyond formal theory rather than providing a direct validation of the limiting formulas.

\subsection{Noisy Observations}\label{sec:noise}

We now consider the case where the data are affected by noise, as is often the case with real observations. We examine the cases where noise is added to both $\mathbf{X}$ and {$\mathbf{y}$}. In the first scenario, we flip a percentage of the response {$\mathbf{y}$} to the opposite sign, which makes some data have a spurious relationship with the response. In the second case, Gaussian noise is added to a percentage of data columns, which increases their variance.

In \figref{fig:noise_Y}, we notice that the flipping response raises the overall error rate but increasing $\gamma$ reduces the error rate. However, when adding varying levels of Gaussian noise to a fixed percentage of data, a higher noise level results in a more prominent increase in $\gamma$, exhibiting a second rise in error (\figref{fig:noise_X}). This suggests that model complexity can diminish the model performance when significant noise is presented in features.
\begin{figure}[!htb]
\centering
\includegraphics[scale = 0.4]{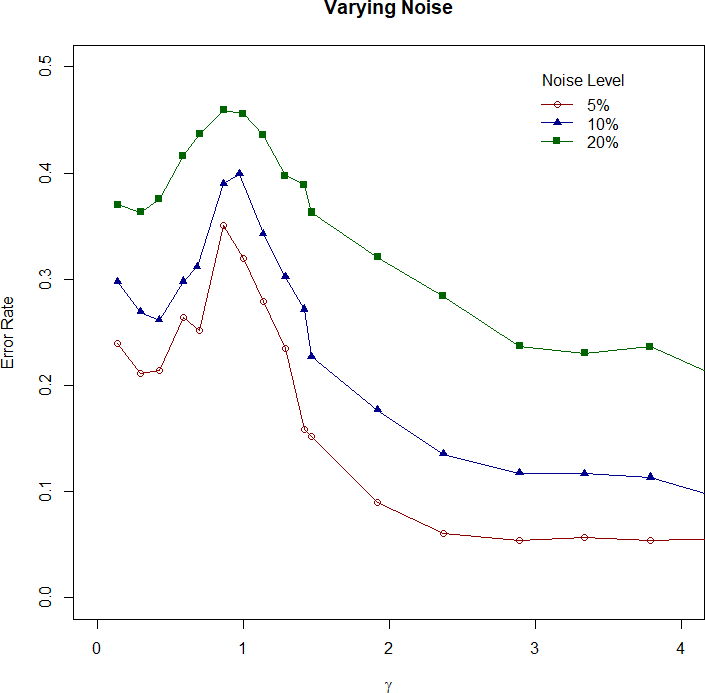}
\caption{Noisy response $Y$; $\boldsymbol{\delta}_p = 1_p$, $\boldsymbol{\Sigma} = \mathbf{I}_p$}
\label{fig:noise_Y}
\end{figure}

\begin{figure}[!htb]
\centering
\includegraphics[scale=0.4]{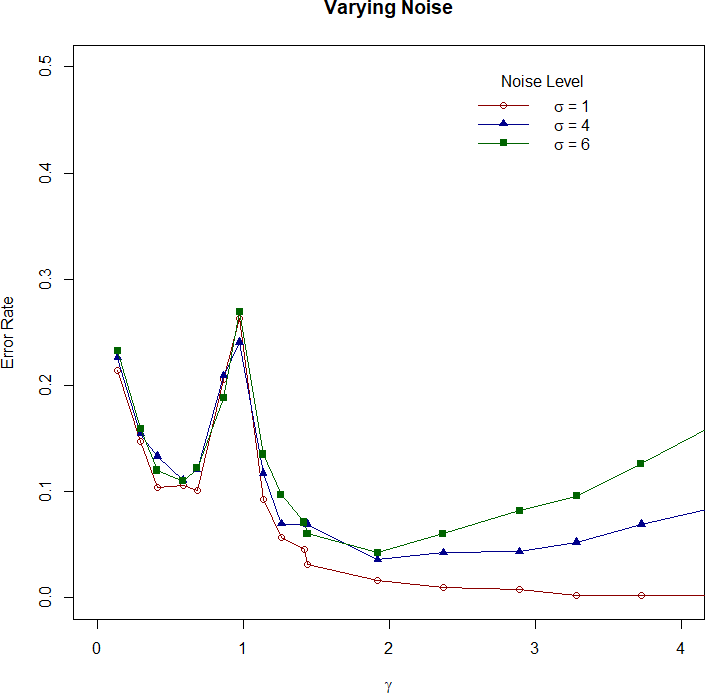}
\caption{Noisy features $ \mathbf{x}_{ki} = \mathbf{x}_{ki} + \boldsymbol{\epsilon}_i$,  $\boldsymbol{\epsilon}_i  \sim \mathcal{N}(\mathbf{0},\sigma^2\mathbf{I}_p)$;$\boldsymbol{\delta}_p = \boldsymbol{1}_p$, $\boldsymbol{\Sigma} = \mathbf{I}_p$}
\label{fig:noise_X}
\end{figure}

\subsection{Redundant Features}\label{sec:redundant}

To mimic real-life high-dimensional data, we generate a percentage of irrelevant features. The result (\figref{fig:irre}) shows that the peak in error has been postponed. This can be explained by the fact that the ``effective'' number of $p$ is smaller to achieve $p \simeq n$. However, if the more redundant feature exists, the error rate is higher at larger $\gamma$ values (in an over-parameterized regime). This observation is in accordance with that noticed by \cite{chatterji2021}: when irrelevant variables increase, the performance of the max-margin classifier degrades with $\gamma$.
\begin{figure}[!htb]
\centering
\includegraphics[scale = 0.4]{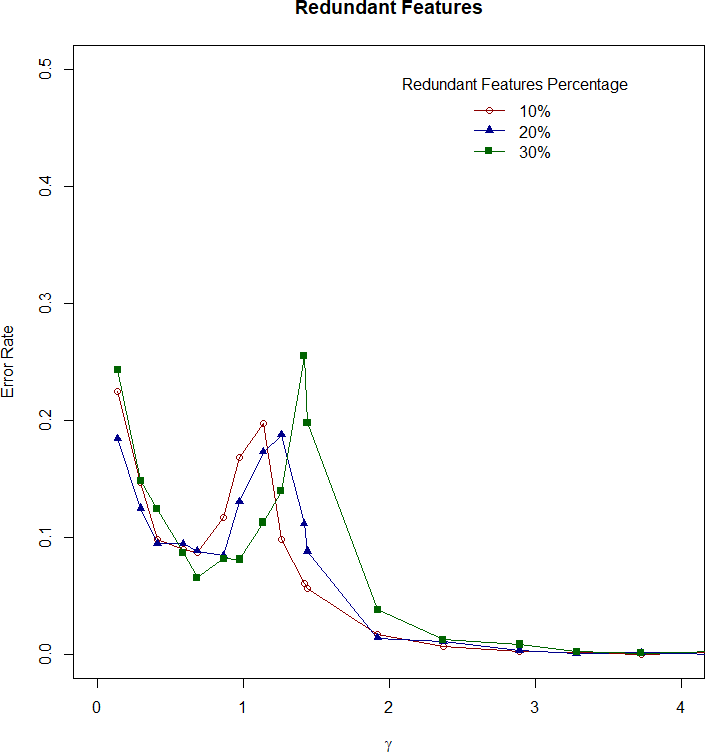}
\caption{Varying the percentage of irrelevant features (${\boldsymbol{\delta}_p} = \boldsymbol{1}_p, \boldsymbol{\Sigma} = \mathbf{I}_p$)}\label{fig:irre}
\end{figure}

In general, we verified the hypothesis that the double descent phenomenon is observed in various settings of $\boldsymbol{\Sigma}$, mild misspecification of the model, and noise. In all cases, the global minimum is attained at the over-parameterized region. However, in noisy features, the error curve in the over-parameterized regime ($\gamma >1$) can quickly ascend. Hence, large over-parameterization is likely not ideal in most real data scenarios.

\section{Empirical Data Analysis}\label{section6}

\subsection{Setup}

In practice, the true population parameters of the data are unknown, the observations can be noisy, and the data may violate various LDA's assumptions. Hence, we perform real data analysis on the ARCENE cancer classification dataset \citep{ARCENE} and compare the performance of pseudo-inverse LDA (PLDA) to dimension reduction methods including PCA, diagonal discriminant analysis (DLDA), shrunken centroid discriminant analysis (SLDA), and maximum uncertainty LDA (MLDA). PCA works by finding the eigen-decomposition of the covariance matrix to identify orthogonal directions of the data where variance is maximized. We fit the LDA model with a varying number of principal components. In DLDA \citep[see, e.g.,][]{dudoit2002}, the off-diagonal elements of the covariance matrix are set to $0$, and the inverse is found by inverting the variances. SLDA utilizes the inverse of an optimally shrunken covariance matrix estimate and is constructed from a convex linear combination of the sample covariance matrix and the identity matrix \citep{Ledoit2004}. MLDA \citep{TKG06} estimates the covariance matrix through the maximum entropy covariance selection approach. These methods have been shown to improve LDA classification performance in a high dimensional setting by better conditioning the covariance matrix inverse \citep{Sharma2015}. Therefore, we hypothesize that they are likely to perform better than LDA with pseudo-inverse in terms of error rate. The result sheds light on the benefit of overparameterization compared to the use of regularized methods in real-world applications.

The goal of the classification task in ARCENE is to distinguish cancer from normal patterns using mass spectrometry data. Mass spectrometry is a technique that identifies and quantifies the chemical composition by measuring the mass-to-charge ratio of molecules. The features are the level of proteins in human Sera for a given mass value. The response is a binary attribute with, {$y = 1$} indicating a cancer-positive patient, and {$y = -1$} indicating a cancer-negative patient.  The original training set consists of $ p =10,000$ features with $n = 100$ samples, and among them, $3000$ of the features are probe features without discriminatory power. All features are numerical and continuous. As part of the exploratory analysis, we conducted model checking. Given the size of the features, we randomly select a subset of features to perform preliminary model checks to understand the data. To compute the error rate, we randomly select a number of features to achieve $\gamma$s between 0 to 10, and repeat the experiment 50 times and average the error rate. The error rate is calculated from $\frac{1}{n_t} \sum_{i=1}^{n_t} \mathds{1}({y_{i} \neq \widehat y_{i}})$ with $n_t = 99$ and the process is repeated 30 times to obtain an average estimate.

The normality assumption underpins LDA, which requires that each variable be normally distributed. Visually inspecting the Q-Q plot of randomly selected attributes, normality is heavily violated for most variables due to the skewed distribution with clustering of data points around $0$. The Mardia skewness and kurtosis test \citep{mardia1970} on a random subset of features also supports the fact that the majority groups of variables fail the multivariate normality assumption. Moreover, the common covariance assumption for LDA has also been violated by conducting the Box's M test on a subset of features. The result shows that the $p$-value $p < 0.05$ for almost all subsets, suggesting that the null hypothesis that the covariances are equal across the groups should be rejected. Since the Mahalanobius distance implies discriminatory ability, we further calculated the Mahalanobius distance of the two groups when varying $\gamma$.~\figref{fig:mahalan_dis} shows the averaged Mahalanobis distance calculated from randomly selected features. The exponential increase in the distance suggests that over-parameterization may help with decreasing the error rate. 

\begin{figure}[!htb]
\centering
\includegraphics[scale = 0.42]{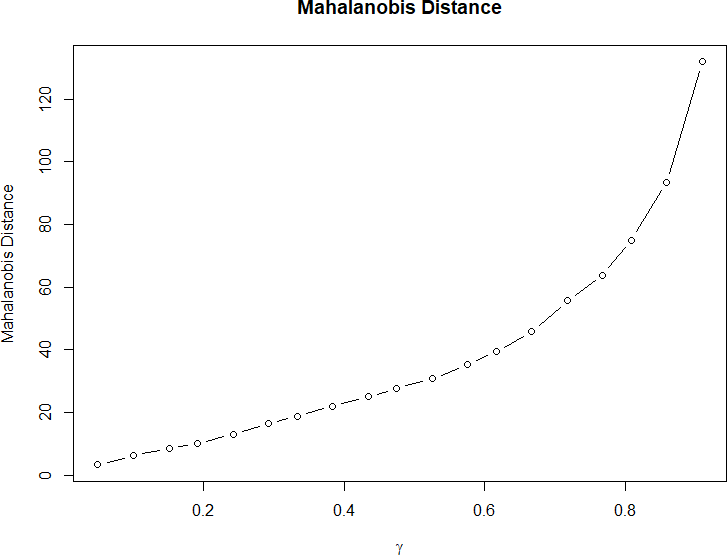}
\caption{Mahalanobis Distance for $\gamma \in (0, 1)$ (Result is averaged over 50 trials of randomly selected features)}
\label{fig:mahalan_dis}
\end{figure}

\subsection{Result}\label{sec:5.2}

\figref{fig:arcene_result} shows an error curve that exhibits double descent. In the under-parameterized regime, the model performance is subpar given the data strongly deviates from the model assumption and the classical $U$-shape error curve is observed. The error peaks at random guessing $0.5$ when $\gamma = 1$, and decreases in the over-parameterized regime when the number of features is increased to $1000$. Compared to other methods, we find that the dimension reduction methods perform similarly to PLDA. However, at large $\gamma$ values, DLDA, SLDA, and MLDA all converge to a high error rate ($\simeq 0.45$), whilst PLDA achieves a lower error (\figref{fig:comparison} and Table~\ref{tab:comparison}). In \figref{fig:PCA}, we also find that PCA achieves minimum error with around 60 principal components, but as $p$ approaches $n$, PCA has the highest error rate relative to methods like MLDA, due to the ill-conditioning of the covariance matrices.
\begin{figure}[!htb]
\centering
\includegraphics[scale = 0.37]{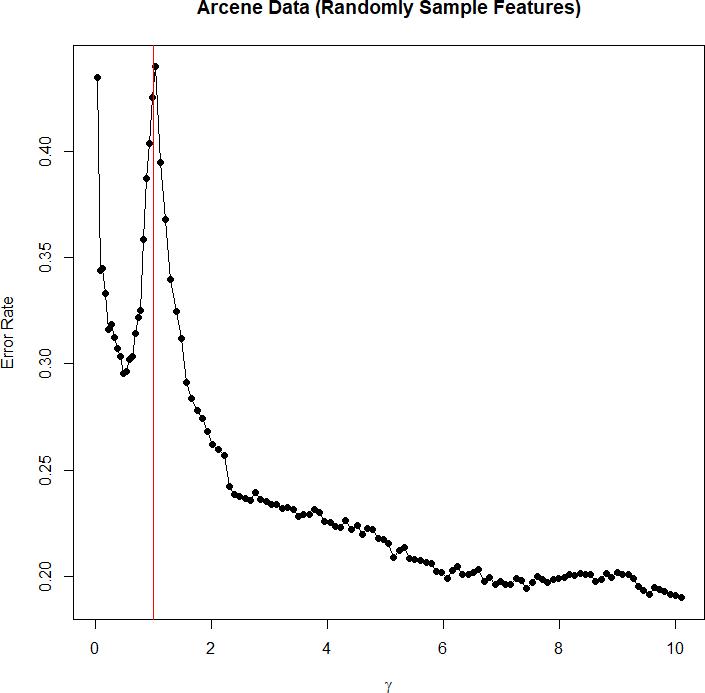}
\caption{ARCENE data observed error rate for $\gamma \in (0, 10)$. Pseudo-LDA is applied and features are randomly selected with results averaged over 50 trials}\label{fig:arcene_result}
\end{figure}

\begin{figure}[!htb]
\centering
\includegraphics[scale = 0.37]{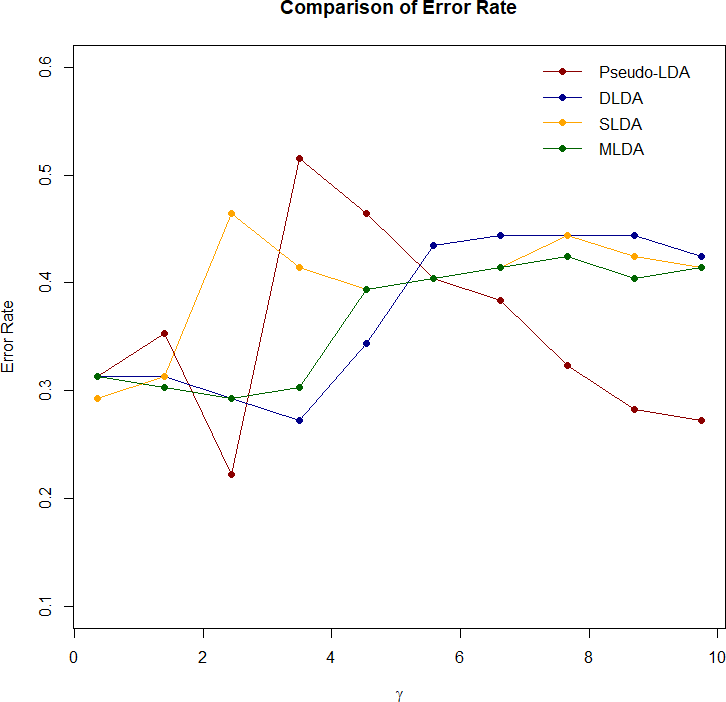}
\caption{Comparison with other methods of high dimensional LDA using originally ordered features $\gamma \in (0, 10)$}\label{fig:comparison}
\end{figure}

Using originally ordered features and comparing PLDA error rates to other high dimensional LDA techniques, we can see that the behavior of the error rate is drastically different. In the underparameterized regime, PLDA has the highest error rate relative to methods such as MLDA, due to the ill-conditioning of covariance matrices as $p$ approaches $n$. However, at large $\gamma$ values, DLDA, SLDA and MLDA all converge to a high error rate ($\simeq 0.45$), whilst PLDA has a continuously decreasing error rate after a certain point. This is confirmed by the table of minimum error, pseudo-inverse's performance is most competitive in the over-parameterized regime. The overall lowest error is also achieved by PLDA.
\begin{table}[!htb]
\centering
\begin{tabular}{|c|c|c|} \hline 
&  Under-parameterized Minimum Error &  Over-parameterized Minimum Error \\ \hline 
PLDA&  0.293 ($\gamma$ = 0.7)&  0.202 ($\gamma$ = 2.4)\\ \hline 
DLDA&  0.303 ($\gamma$= 0.7)&  0.273 ($\gamma$ = 3.5)\\ \hline 
SLDA&  0.293 ($\gamma$ = 0.4)&  0.253 ($\gamma$ = 2.1)\\ \hline
MLDA& 0.212 ($\gamma$= 0.7)& 0.253 ($\gamma$ = 1.1)\\\hline
\end{tabular}
\caption{Error rate comparison for different High Dimensional LDA methods}\label{tab:comparison}
\end{table}

\begin{figure}[!htb]
\centering
\includegraphics[scale = 0.33]{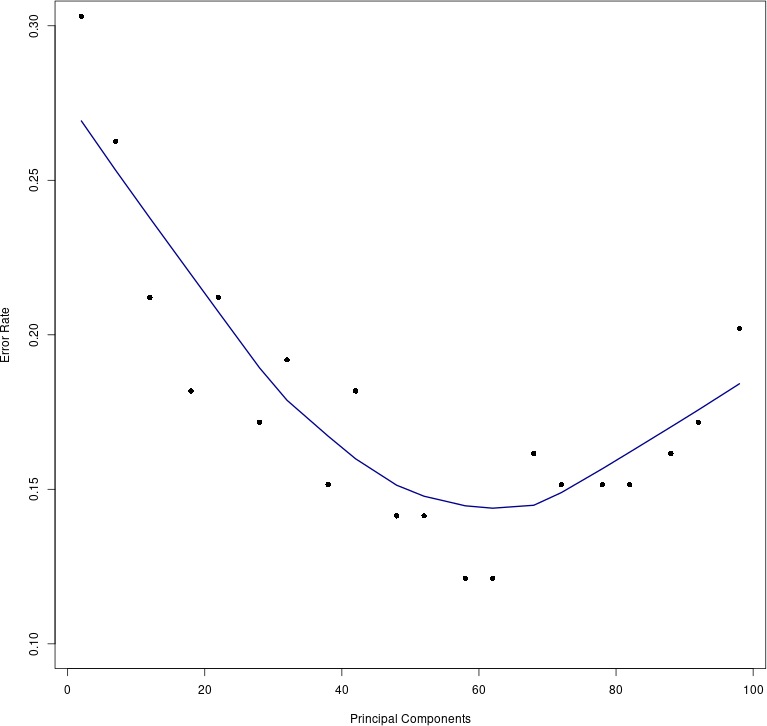}
\caption{LDA Error Rate on ARCENE data with features as principal components}\label{fig:PCA}
\end{figure}

\section{Conclusion}\label{sec:conclusion}

This paper studies the mis-classification risk of linear discriminant analysis in the proportional-growth regime $p/n\to\gamma$. Our results show that the risk can depend non-monotonically on the dimensionality-to-sample-size ratio, producing a double-descent pattern even for this classical linear method. Under the homoscedastic Gaussian model, Theorem~\ref{Thm1} gives an explicit limiting error formula in the under-parameterized regime for a general deterministic positive-definite covariance matrix. Under isotropic covariance, Theorem~\ref{Thm2} extends the analysis to the over-parameterized regime for the Moore--Penrose pseudoinverse LDA classifier. In the common isotropic setting, the two limits join continuously at $\gamma=1$. The limiting formulas separate the effects of mean and covariance estimation and clarify how the dimensional ratio and the Mahalanobis distance jointly determine the classification performance. The simulations support the theoretical predictions in the settings covered by the theory and examine whether similar patterns persist beyond it, under more general covariance structures, non-Gaussian feature distributions, and label and feature noise. The ARCENE analysis further illustrates that non-monotone risk behavior can arise in a traditional statistical classification procedure. Together, these findings show that double descent is not exclusive to highly over-parameterized nonlinear models. To facilitate reproducibility, the \Rlogo\ code used to generate all figures is available at \url{https://github.com/hanshang/benign_overfitting}.

Several directions merit further investigation. The most immediate is to establish the over-parameterized limit under a general covariance structure, since Theorem~\ref{Thm2} currently covers only the isotropic case. The second direction is to relax the Gaussian assumption to weaker distributional conditions. Incorporating regularized or shrinkage-based LDA and allowing class-dependent covariance structures are also useful extensions. More broadly, our results suggest that random matrix theory provides a useful framework to understand the risk of high-dimensional classification methods.

\section*{Acknowledgment}

The authors thank the insightful comments of the two reviewers. We are also grateful for financial support from an Australian Research Council Discovery Project (DP230102250), and to Mr. Ziyue (Humphrey) Yang from the Australian National University for his help with the \Rlogo \ code. 

\newpage
\appendix

\renewcommand{\thefigure}{\thesection.\arabic{figure}}
\setcounter{figure}{0}
\section{Theoretical Proofs}\label{appendix1}

We collect here the notation used in the proofs. Throughout, \(\|\boldsymbol{A}\|\) denotes the spectral norm of a matrix \(\boldsymbol{A}\), and \(\|\boldsymbol{a}\|\) denotes the Euclidean norm of a vector \(\boldsymbol{a}\). We write \(a_n\asymp b_n\) if there exist constants \(0<c<C<\infty\) such that \(c\,b_n\le a_n\le C\,b_n\) for all sufficiently large \(n\). Likewise, \(a_n\lesssim b_n\) means \(a_n\le C\,b_n\) for some constant \(C>0\), and \(a_n\gtrsim b_n\) means \(a_n\ge c\,b_n\) for some constant \(c>0\). For symmetric matrices \(\boldsymbol{A}\) and \(\boldsymbol{B}\), the notation \(\boldsymbol{A}\succeq \boldsymbol{B}\) means that \(\boldsymbol{A}-\boldsymbol{B}\) is positive semidefinite.

\begin{proof}[Proof of Theorem~\ref{Thm1}]
Throughout the proof we work under Assumptions~\ref{ASSU:ratio}--\ref{ASSU:signal} with $\gamma\in(0,1)$. The proof starts from the exact conditional error formula in~\eqref{eq:lda_error}. Since the observations are Gaussian, conditional on the training sample the LDA score of an independent test point is Gaussian, so it remains only to find the asymptotic limit of the standardized score inside $\Phi$: $R_{\mathrm{LDA}}=\frac12\sum_{k=1}^{2}\Phi(A_k)$, where
\[
A_k:=
\frac{(-1)^k\left(\boldsymbol{\mu}_k-\frac{\widehat{\boldsymbol{\mu}}_1+\widehat{\boldsymbol{\mu}}_2}{2}\right)^{\top}\widehat{\boldsymbol{\Sigma}}^{-1}(\widehat{\boldsymbol{\mu}}_1-\widehat{\boldsymbol{\mu}}_2)}
{\left[(\widehat{\boldsymbol{\mu}}_1-\widehat{\boldsymbol{\mu}}_2)^{\top}\widehat{\boldsymbol{\Sigma}}^{-1}\boldsymbol{\Sigma}\widehat{\boldsymbol{\Sigma}}^{-1}(\widehat{\boldsymbol{\mu}}_1-\widehat{\boldsymbol{\mu}}_2)\right]^{1/2}},
\qquad k=1,2.
\]
By the balanced design $n_1=n_2$, the two class arguments have the same limit, so we focus on $A_1$.

\smallskip\noindent\underline{\textit{Whitening and reduction.}}
Let $\boldsymbol{b}:=\boldsymbol{\Sigma}^{-1/2}(\boldsymbol{\mu}_1-\boldsymbol{\mu}_2)$, so that $\|\boldsymbol{b}\|=\Delta$, and define the whitened class-mean errors $\mathbf a_k:=\boldsymbol{\Sigma}^{-1/2}(\widehat{\boldsymbol{\mu}}_k-\boldsymbol{\mu}_k)$, $k=1,2$. Under the Gaussian model, $\mathbf a_k\sim\mathcal N(\mathbf 0,n_k^{-1}\mathbf I_p)$, and the class sample means are independent of the pooled within-class covariance. Since $n_1=n_2=n/2$, we have $\mathbf a_k\sim\mathcal N(\mathbf 0,2n^{-1}\mathbf I_p)$. Setting $\mathbf a_-:=\mathbf a_1-\mathbf a_2$ and $\mathbf a_+:=\mathbf a_1+\mathbf a_2$, the equal-size design gives $\operatorname{Cov}(\mathbf a_+,\mathbf a_-)=\operatorname{Var}(\mathbf a_1)-\operatorname{Var}(\mathbf a_2)=\mathbf 0$, and under the Gaussian model zero covariance implies independence, so $\mathbf a_+$ and $\mathbf a_-$ are independent centered Gaussian vectors, each with covariance $4n^{-1}\mathbf I_p$.

Stack the whitened centered observations $\mathbf z_{ki}:=\boldsymbol{\Sigma}^{-1/2}(\mathbf x_{ki}-\boldsymbol{\mu}_k)$ into an $n\times p$ matrix $\mathbf Z$. The within-class demeaning matrix is
\[
\mathbf H=
\begin{pmatrix}
\mathbf I_{n_1}-n_1^{-1}\mathbf{1}_{n_1}\mathbf{1}_{n_1}^{\top} & 0\\
0 & \mathbf I_{n_2}-n_2^{-1}\mathbf{1}_{n_2}\mathbf{1}_{n_2}^{\top}
\end{pmatrix},
\]
where $\mathbf{1}_m$ denotes the $m$-dimensional vector of ones, so that $\widehat{\boldsymbol{\Sigma}}=n^{-1}\boldsymbol{\Sigma}^{1/2}\mathbf Z^{\top}\mathbf H\mathbf Z\boldsymbol{\Sigma}^{1/2}$. Define the whitened pooled covariance $\mathbf S:=\boldsymbol{\Sigma}^{-1/2}\widehat{\boldsymbol{\Sigma}}\boldsymbol{\Sigma}^{-1/2}=n^{-1}\mathbf Z^{\top}\mathbf H\mathbf Z$, so that $\boldsymbol{\Sigma}^{-1/2}(\widehat{\boldsymbol{\mu}}_1-\widehat{\boldsymbol{\mu}}_2)=\boldsymbol{b}+\mathbf a_-$. The vector $\mathbf a_-$ has covariance $4n^{-1}\mathbf I_p$, which is the source of the $4\gamma$ term in the final limit. After whitening, the class-$1$ numerator becomes $\frac12(\boldsymbol{b}-\mathbf a_+)^{\top}\mathbf S^{-1}(\boldsymbol{b}+\mathbf a_-)$ and the squared denominator becomes $(\boldsymbol{b}+\mathbf a_-)^{\top}\mathbf S^{-2}(\boldsymbol{b}+\mathbf a_-)$. Expanding, define
\[
D_1:=\boldsymbol{b}^{\top}\mathbf S^{-1}\boldsymbol{b},\quad
D_2:=\boldsymbol{b}^{\top}\mathbf S^{-1}\mathbf a_-,\quad
D_3:=\mathbf a_+^{\top}\mathbf S^{-1}\boldsymbol{b},\quad
D_4:=\mathbf a_+^{\top}\mathbf S^{-1}\mathbf a_-,
\]
\[
D_5:=\boldsymbol{b}^{\top}\mathbf S^{-2}\boldsymbol{b},\quad
D_6:=\mathbf a_-^{\top}\mathbf S^{-2}\mathbf a_-,\quad
D_7:=\boldsymbol{b}^{\top}\mathbf S^{-2}\mathbf a_-,
\]
so that
\begin{equation}\label{Eqn:A_in_Dterms}
A_1=-\frac{\frac12(D_1+D_2-D_3-D_4)}{\sqrt{D_5+2D_7+D_6}},
\end{equation}
where the minus sign comes from $(-1)^1=-1$.

\smallskip\noindent\underline{\textit{Spectral properties of $\mathbf S$.}}
Since $\mathbf Z$ has i.i.d.\ standard Gaussian entries and $\mathbf H$ is a deterministic projection of rank $n-2$, the matrix $\mathbf Z^{\top}\mathbf H\mathbf Z$ has the Wishart distribution $W_p(n-2,\mathbf I_p)$; in particular, the eigenvector matrix of $\mathbf S$ is Haar distributed and independent of the eigenvalues, and the normalization by $n$ rather than $n-2$ is asymptotically immaterial. Hence the empirical spectral distribution of $\mathbf S$ converges to the Mar\v{c}enko--Pastur law with aspect ratio $\gamma$, and by the Bai--Yin theorem $\lambda_{\min}(\mathbf S)\convergeInP(1-\sqrt\gamma)^{2}>0$, so the spectrum of $\mathbf S$ is bounded away from zero with probability tending to one. Consequently $\mathbf S$ is non-singular with probability tending to one, with $\widehat{\boldsymbol{\Sigma}}^{-1}=\boldsymbol{\Sigma}^{-1/2}\mathbf S^{-1}\boldsymbol{\Sigma}^{-1/2}$ and $\widehat{\boldsymbol{\Sigma}}^{-1}\boldsymbol{\Sigma}\widehat{\boldsymbol{\Sigma}}^{-1}=\boldsymbol{\Sigma}^{-1/2}\mathbf S^{-2}\boldsymbol{\Sigma}^{-1/2}$, and combining the spectral-gap bound with the Mar\v{c}enko--Pastur convergence gives
\[
\frac1p\operatorname{tr}(\mathbf S^{-1})\convergeInP\frac1{1-\gamma},
\qquad
\frac1p\operatorname{tr}(\mathbf S^{-2})\convergeInP\frac1{(1-\gamma)^3},
\qquad
\operatorname{tr}(\mathbf S^{-j})=O_p(p),\quad j=1,\ldots,4.
\]

\smallskip\noindent\underline{\textit{Signal terms $D_1$ and $D_5$.}}
Since the eigenvectors of $\mathbf S$ are Haar distributed and independent of the eigenvalues, standard concentration of Haar quadratic forms gives, for the deterministic unit vector $\boldsymbol{\nu}:=\boldsymbol{b}/\|\boldsymbol{b}\|$,
\[
\boldsymbol{\nu}^{\top}\mathbf S^{-1}\boldsymbol{\nu}-\frac1p\operatorname{tr}(\mathbf S^{-1})\convergeInP0,
\qquad
\boldsymbol{\nu}^{\top}\mathbf S^{-2}\boldsymbol{\nu}-\frac1p\operatorname{tr}(\mathbf S^{-2})\convergeInP0.
\]
Together with the trace limits above and $\|\boldsymbol{b}\|^{2}=\Delta^{2}=O(1)$ by Assumption~\ref{ASSU:signal},
\[
D_1-\frac{\Delta^{2}}{1-\gamma}\convergeInP0,
\qquad
D_5-\frac{\Delta^{2}}{(1-\gamma)^{3}}\convergeInP0.
\]

\smallskip\noindent\underline{\textit{Cross terms $D_2$, $D_3$, $D_7$.}}
The Gaussian assumption gives independence between $\mathbf S$ and the sample-mean errors. Conditional on $\mathbf S$, $\mathbb E(D_2\mid\mathbf S)=0$ and $\operatorname{Var}(D_2\mid\mathbf S)=\frac4n\boldsymbol{b}^{\top}\mathbf S^{-2}\boldsymbol{b}=\frac4nD_5=O_p(n^{-1})$, so $D_2=o_p(1)$; the same argument gives $D_3=o_p(1)$. For $D_7$, $\mathbb E(D_7\mid\mathbf S)=0$ and $\operatorname{Var}(D_7\mid\mathbf S)=\frac4n\boldsymbol{b}^{\top}\mathbf S^{-4}\boldsymbol{b}\le\frac4n\lambda_{\min}(\mathbf S)^{-4}\Delta^{2}=O_p(n^{-1})$, so $D_7=o_p(1)$.

\smallskip\noindent\underline{\textit{Cross term $D_4$.}}
Conditional on $\mathbf S$, $\mathbf a_+$ and $\mathbf a_-$ are independent and centered, so $\mathbb E(D_4\mid\mathbf S)=0$ and $\operatorname{Var}(D_4\mid\mathbf S)=\frac{16}{n^{2}}\operatorname{tr}(\mathbf S^{-2})=O_p(p/n^{2})=O_p(n^{-1})$, and therefore $D_4=o_p(1)$.

\smallskip\noindent\underline{\textit{Noise term $D_6$.}}
Conditional on $\mathbf S$,
\[
\mathbb E(D_6\mid\mathbf S)=\frac4n\operatorname{tr}(\mathbf S^{-2})\convergeInP\frac{4\gamma}{(1-\gamma)^3},
\qquad
\operatorname{Var}(D_6\mid\mathbf S)=2\Bigl(\frac4n\Bigr)^2\operatorname{tr}(\mathbf S^{-4})=O_p(p/n^{2})=o_p(1),
\]
using $p/n\to\gamma$, $p^{-1}\operatorname{tr}(\mathbf S^{-2})\convergeInP(1-\gamma)^{-3}$ and $\operatorname{tr}(\mathbf S^{-4})=O_p(p)$. Thus $D_6\convergeInP4\gamma/(1-\gamma)^{3}$.

\smallskip\noindent\underline{\textit{Conclusion.}}
Substituting these limits into~\eqref{Eqn:A_in_Dterms}, the numerator satisfies
\[
-\frac12(D_1+D_2-D_3-D_4)+\frac{\Delta^{2}}{2(1-\gamma)}\convergeInP0,
\]
while the squared denominator satisfies
\[
D_5+2D_7+D_6-\frac{\Delta^{2}+4\gamma}{(1-\gamma)^{3}}\convergeInP0.
\]
Since the limit of the squared denominator is bounded below by $4\gamma/(1-\gamma)^{3}>0$, Slutsky's theorem gives
\[
A_1+\frac{\sqrt{1-\gamma}\,\Delta^{2}}{2\sqrt{\Delta^{2}+4\gamma}}\convergeInP0.
\]
By the symmetry $n_1=n_2$, the class-2 argument $A_2$ has the same asymptotic limit, and since $\Phi$ is Lipschitz continuous,
\[
R_{\mathrm{LDA}}-\Phi\!\left(-\frac{\sqrt{1-\gamma}\,\Delta^{2}}{2\sqrt{\Delta^{2}+4\gamma}}\right)\convergeInP0.
\]
This completes the proof of Theorem~\ref{Thm1}.
\end{proof}

\begin{proof}[Proof of Theorem~\ref{Thm2}]
Throughout the proof, we work under Assumptions~\ref{ASSU:ratio}--\ref{ASSU:signal} with $\gamma\in(1,\infty)$ and $\boldsymbol{\Sigma}=\sigma^2\mathbf I_p$. We adopt the notation from the proof of Theorem~\ref{Thm1}: $\boldsymbol{b}$, $\mathbf a_k$, $\mathbf a_{\pm}$, $\mathbf Z$, $\mathbf H$, and $\mathbf S=n^{-1}\mathbf Z^{\top}\mathbf H\mathbf Z$ retain their definitions and distributional properties. In particular, $\mathbf a_+$ and $\mathbf a_-$ are independent centered Gaussian vectors with covariance $4n^{-1}\mathbf I_p$ and are independent of $\mathbf S$. Moreover, $\mathbf S$ is orthogonally invariant, so conditional on its eigenvalues, standard concentration results for Haar quadratic forms apply to spectral functions of $\mathbf S$, including $\mathbf S^{+}$ and $(\mathbf S^{+})^2$. As before, $R_{\mathrm{LDA}}=\frac12\sum_{k=1}^{2}\Phi(A_k)$, where $A_k$ is defined as in the proof of Theorem~\ref{Thm1} with $\widehat{\boldsymbol{\Sigma}}^{-1}$ replaced by $\widehat{\boldsymbol{\Sigma}}^{+}$. We focus on $A_1$.

\smallskip\noindent\underline{\textit{Whitening and reduction.}}
In the isotropic case $\mathbf S=\sigma^{-2}\widehat{\boldsymbol{\Sigma}}$, and the Moore--Penrose inverse satisfies $(\sigma^{2}\mathbf S)^{+}=\sigma^{-2}\mathbf S^{+}$. Hence $\widehat{\boldsymbol{\Sigma}}^{+}=\boldsymbol{\Sigma}^{-1/2}\mathbf S^{+}\boldsymbol{\Sigma}^{-1/2}$ and $\widehat{\boldsymbol{\Sigma}}^{+}\boldsymbol{\Sigma}\widehat{\boldsymbol{\Sigma}}^{+}=\boldsymbol{\Sigma}^{-1/2}(\mathbf S^{+})^{2}\boldsymbol{\Sigma}^{-1/2}$, so all $\sigma$ factors cancel from $A_1$ and the decomposition~\eqref{Eqn:A_in_Dterms} carries over with $\mathbf S^{-1}$ and $\mathbf S^{-2}$ replaced by $\mathbf S^{+}$ and $(\mathbf S^{+})^{2}$:
\begin{equation}\label{Eqn:A_in_Dterms_over}
A_1=-\frac{\frac12(D_1^{+}+D_2^{+}-D_3^{+}-D_4^{+})}{\sqrt{D_5^{+}+2D_7^{+}+D_6^{+}}},
\end{equation}
where
\[
D_1^{+}:=\boldsymbol{b}^{\top}\mathbf S^{+}\boldsymbol{b},\quad
D_2^{+}:=\boldsymbol{b}^{\top}\mathbf S^{+}\mathbf a_-,\quad
D_3^{+}:=\mathbf a_+^{\top}\mathbf S^{+}\boldsymbol{b},\quad
D_4^{+}:=\mathbf a_+^{\top}\mathbf S^{+}\mathbf a_-,
\]
\[
D_5^{+}:=\boldsymbol{b}^{\top}(\mathbf S^{+})^{2}\boldsymbol{b},\quad
D_6^{+}:=\mathbf a_-^{\top}(\mathbf S^{+})^{2}\mathbf a_-,\quad
D_7^{+}:=\boldsymbol{b}^{\top}(\mathbf S^{+})^{2}\mathbf a_-.
\]

\smallskip\noindent\underline{\textit{Spectral properties of $\mathbf S^{+}$.}}
Since $\gamma>1$, $\mathbf S$ is singular, with $p-(n-2)$ zero eigenvalues, and the empirical spectral distribution of $\mathbf S$ converges to the Mar\v{c}enko--Pastur law $F_\gamma$, which places mass $1-\gamma^{-1}$ at the origin and mass $\gamma^{-1}$ on the bulk component supported on $[(\sqrt\gamma-1)^{2},(\sqrt\gamma+1)^{2}]$. By the Bai--Yin theorem, the smallest nonzero eigenvalue satisfies $\lambda_{\min}^{+}(\mathbf S)\convergeInP(\sqrt\gamma-1)^{2}>0$, so the nonzero spectrum of $\mathbf S$ is bounded away from zero with probability tending to one and $\|\mathbf S^{+}\|=[\lambda_{\min}^{+}(\mathbf S)]^{-1}=O_p(1)$. Combining this spectral gap with the Mar\v{c}enko--Pastur convergence gives
\[
\frac1p\operatorname{tr}(\mathbf S^{+})\convergeInP\int_{x>0}\frac{1}{x}\,dF_\gamma(x)=\frac{1}{\gamma(\gamma-1)},
\qquad
\frac1p\operatorname{tr}\bigl((\mathbf S^{+})^{2}\bigr)\convergeInP\int_{x>0}\frac{1}{x^{2}}\,dF_\gamma(x)=\frac{1}{(\gamma-1)^{3}},
\]
together with $\operatorname{tr}\bigl((\mathbf S^{+})^{j}\bigr)=O_p(p)$ for $j=1,\ldots,4$. The two integrals follow from the companion Stieltjes transform $\underline{m}(z)$ of $F_\gamma$: in the isotropic case the Mar\v{c}enko--Pastur equation gives $\lim_{z\uparrow0}\underline{m}(z)=(\gamma-1)^{-1}$, and expanding $\underline{m}(z)$ to first order at the origin yields the two values above.

\smallskip\noindent\underline{\textit{Signal terms $D_1^{+}$ and $D_5^{+}$.}}
On the event that the nonzero spectrum of $\mathbf S$ is bounded away from zero, the spectral functions $x^{-1}\mathbf 1\{x>0\}$ and $x^{-2}\mathbf 1\{x>0\}$ are uniformly bounded. Since $\mathbf S$ is orthogonally invariant, conditional on its eigenvalues, standard concentration of Haar quadratic forms applies to the uniquely defined spectral functions $\mathbf S^{+}$ and $(\mathbf S^{+})^2$. Therefore, for $\boldsymbol{\nu}:=\boldsymbol{b}/\|\boldsymbol{b}\|$,
\[
\boldsymbol{\nu}^{\top}\mathbf S^{+}\boldsymbol{\nu}
-\frac1p\operatorname{tr}(\mathbf S^{+})
\convergeInP0,
\qquad
\boldsymbol{\nu}^{\top}(\mathbf S^{+})^2\boldsymbol{\nu}
-\frac1p\operatorname{tr}\bigl((\mathbf S^{+})^2\bigr)
\convergeInP0.
\]

Together with the trace limits above and $\Delta^{2}=O(1)$,
\[
D_1^{+}-\frac{\Delta^{2}}{\gamma(\gamma-1)}\convergeInP0,
\qquad
D_5^{+}-\frac{\Delta^{2}}{(\gamma-1)^{3}}\convergeInP0.
\]

\smallskip\noindent\underline{\textit{Cross terms $D_2^{+}$, $D_3^{+}$, $D_7^{+}$.}}
Conditional on $\mathbf S$, $\mathbb E(D_2^{+}\mid\mathbf S)=0$ and $\operatorname{Var}(D_2^{+}\mid\mathbf S)=\frac4n\boldsymbol{b}^{\top}(\mathbf S^{+})^{2}\boldsymbol{b}=\frac4nD_5^{+}=O_p(n^{-1})$, so $D_2^{+}=o_p(1)$; the same argument gives $D_3^{+}=o_p(1)$. For $D_7^{+}$, $\mathbb E(D_7^{+}\mid\mathbf S)=0$ and $\operatorname{Var}(D_7^{+}\mid\mathbf S)=\frac4n\boldsymbol{b}^{\top}(\mathbf S^{+})^{4}\boldsymbol{b}\le\frac4n\|\mathbf S^{+}\|^{4}\Delta^{2}=O_p(n^{-1})$, so $D_7^{+}=o_p(1)$.

\smallskip\noindent\underline{\textit{Cross term $D_4^{+}$.}}
Conditional on $\mathbf S$, $\mathbf a_+$ and $\mathbf a_-$ are independent and centered, so $\mathbb E(D_4^{+}\mid\mathbf S)=0$ and $\operatorname{Var}(D_4^{+}\mid\mathbf S)=\frac{16}{n^{2}}\operatorname{tr}\bigl((\mathbf S^{+})^{2}\bigr)=O_p(p/n^{2})=O_p(n^{-1})$, and therefore $D_4^{+}=o_p(1)$.

\smallskip\noindent\underline{\textit{Noise term $D_6^{+}$.}}
Conditional on $\mathbf S$,
\[
\mathbb E(D_6^{+}\mid\mathbf S)=\frac4n\operatorname{tr}\bigl((\mathbf S^{+})^{2}\bigr)\convergeInP\frac{4\gamma}{(\gamma-1)^{3}},
\qquad
\operatorname{Var}(D_6^{+}\mid\mathbf S)=2\Bigl(\frac4n\Bigr)^{2}\operatorname{tr}\bigl((\mathbf S^{+})^{4}\bigr)=o_p(1),
\]
using $p/n\to\gamma$, $p^{-1}\operatorname{tr}((\mathbf S^{+})^{2})\convergeInP(\gamma-1)^{-3}$ and $\operatorname{tr}((\mathbf S^{+})^{4})=O_p(p)$. Thus $D_6^{+}\convergeInP4\gamma/(\gamma-1)^{3}$.

\smallskip\noindent\underline{\textit{Conclusion.}}
Substituting these limits into~\eqref{Eqn:A_in_Dterms_over}, the numerator satisfies
\[
-\frac12(D_1^{+}+D_2^{+}-D_3^{+}-D_4^{+})+\frac{\Delta^{2}}{2\gamma(\gamma-1)}\convergeInP0,
\]
while the squared denominator satisfies
\[
D_5^{+}+2D_7^{+}+D_6^{+}-\frac{\Delta^{2}+4\gamma}{(\gamma-1)^{3}}\convergeInP0.
\]
Since the limit of the squared denominator is bounded below by $4\gamma/(\gamma-1)^{3}>0$, Slutsky's theorem gives
\[
A_1+\frac{\Delta^{2}/\bigl(2\gamma(\gamma-1)\bigr)}{\sqrt{(\Delta^{2}+4\gamma)/(\gamma-1)^{3}}}
=A_1+\frac{\sqrt{\gamma-1}\,\Delta^{2}}{2\gamma\,\sqrt{\Delta^{2}+4\gamma}}\convergeInP0.
\]
By the symmetry $n_1=n_2$, the class-2 argument $A_2$ has the same asymptotic limit, and since $\Phi$ is Lipschitz continuous,
\[
R_{\mathrm{LDA}}-\Phi\!\left(-\frac{\sqrt{\gamma-1}\,\Delta^{2}}{2\gamma\,\sqrt{\Delta^{2}+4\gamma}}\right)\convergeInP0.
\]
This completes the proof of Theorem~\ref{Thm2}.
\end{proof}

\section{Error Rate under the Normal Scenario}\label{appendix:B}

The error rate under the normal assumption for the under-parameterized region \citep[Theorem~1 of][]{Wang2018} is visualized in \figref{fig:5.2}.
\begin{figure}[!htb]
\begin{subfigure}{0.495\textwidth}
\includegraphics[width=\textwidth]{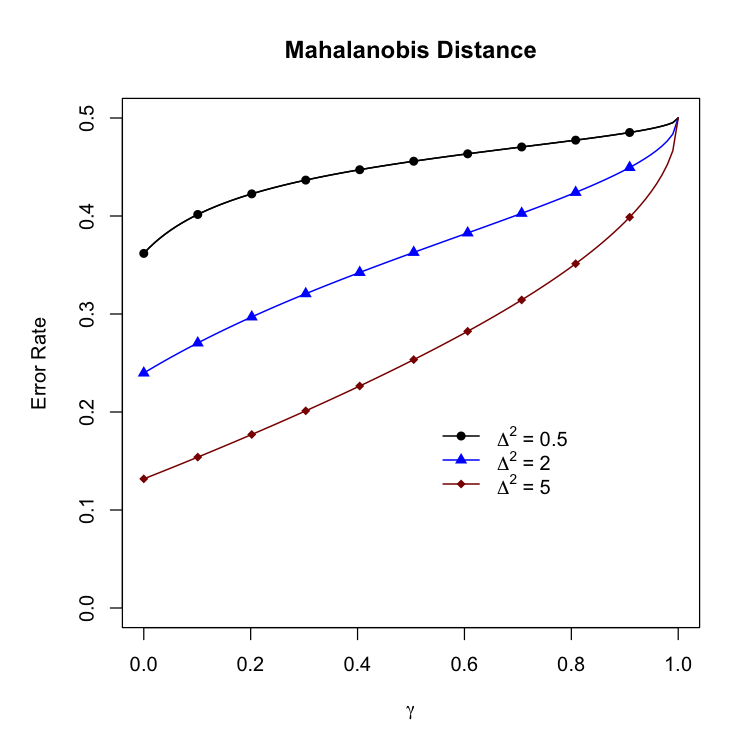}
\caption{Relationship of $\gamma$ and error for varying levels of squared Mahalanobis Distance $\Delta^2$}\label{fig:5.21}
\end{subfigure}
\hfill
\begin{subfigure}{0.495\textwidth}
\includegraphics[width=\textwidth]{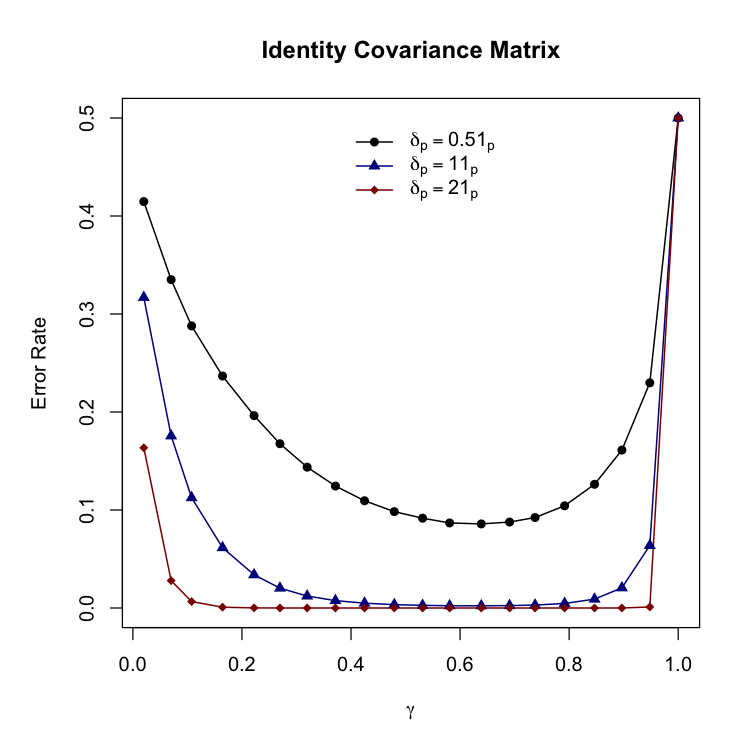}
\caption{Identity covariance matrix with varying population group mean differences. $\boldsymbol{\delta}_p = \boldsymbol{\mu}_1 - \boldsymbol{\mu}_2, \boldsymbol{\Sigma} = \boldsymbol{I}$}
\label{fig:5.22}
\end{subfigure}
\medskip
\begin{subfigure}{0.495\textwidth}
\includegraphics[width=\textwidth]{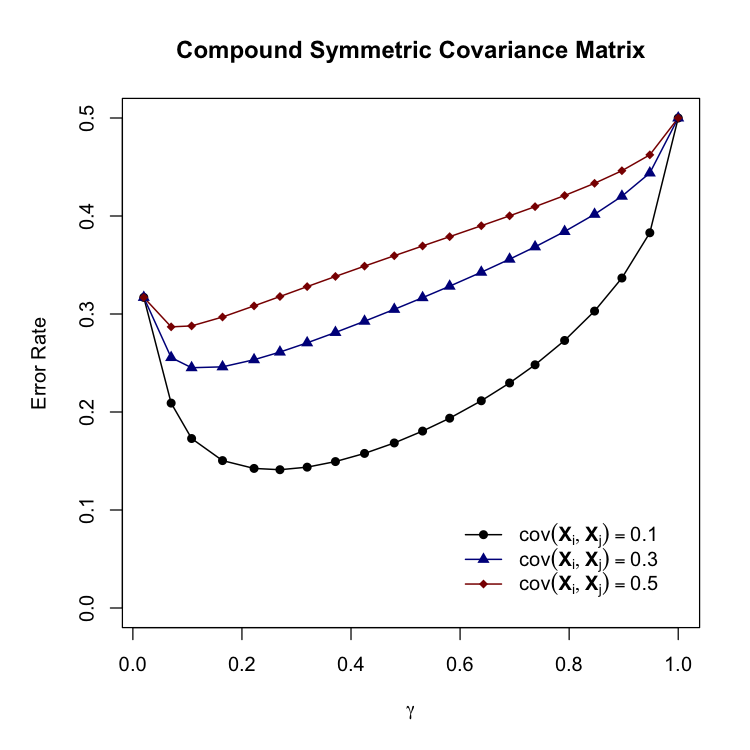}
\caption{Compound symmetric covariance matrix with different covariances}
\label{fig:5.23}
\end{subfigure}
\hfill
\begin{subfigure}{0.495\textwidth}
\includegraphics[width=\textwidth]{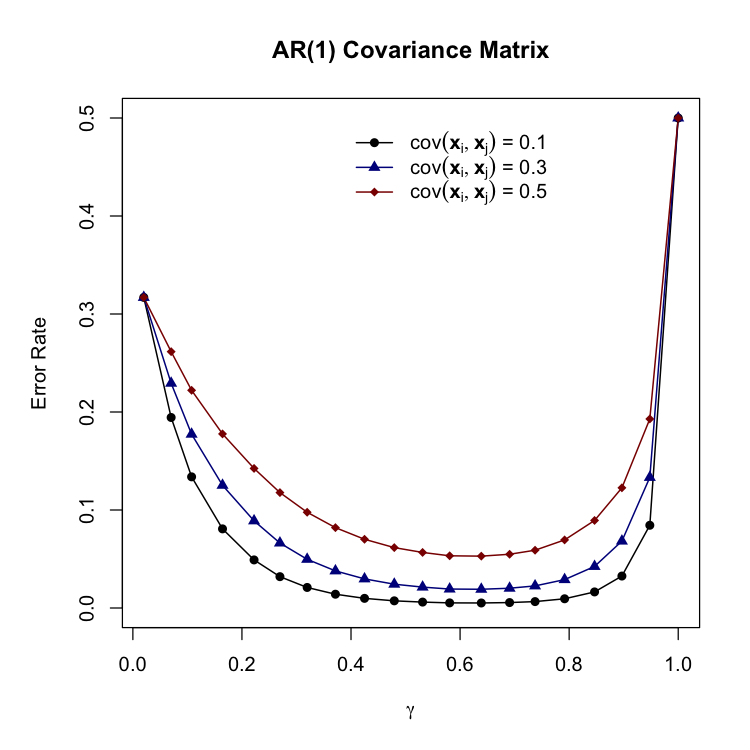}
\caption{AR(1) Covariance matrices with different covariances}
\label{fig:5.24}
\end{subfigure}
\caption{Theoretical model error in under-parameterized regime}
\label{fig:5.2}
\end{figure}

\newpage
\singlespacing
\bibliographystyle{agsm}
\bibliography{References.bib}

\end{document}